\documentclass{article}

\usepackage[margin=1in]{geometry}
\usepackage{graphicx}%
\usepackage{multirow}%
\usepackage{amsmath,amssymb,amsfonts}%
\usepackage{amsthm}%
\usepackage{mathrsfs}%
\usepackage[title]{appendix}%
\usepackage{xcolor}%
\usepackage{textcomp}%
\usepackage{manyfoot}%
\usepackage{booktabs}%
\usepackage{listings}%

\usepackage{enumerate}
\usepackage{amsthm}
\usepackage{amssymb}
\usepackage{graphicx}
\usepackage{tikz}
\usepackage{pgfplots}
\usepgfplotslibrary{fillbetween}
\usetikzlibrary{arrows}
\usepackage{caption}
\usepackage{subcaption}
\usepackage{bbm}
\usepackage{enumitem}
\usepackage{etoolbox}
\usepackage{float}
\usepackage{thmtools, thm-restate}
\usepackage[autostyle, english = american]{csquotes}
\MakeOuterQuote{"}

\usepackage{./lp}

\newtoggle{notblinded}
\settoggle{notblinded}{true}

\newtoggle{draft}
\settoggle{draft}{false}

\usetikzlibrary{calc}
\usetikzlibrary{decorations.pathreplacing}
\usetikzlibrary{shapes, patterns, decorations, fit, intersections, arrows, automata, positioning}
\usepackage{thmtools} 
\usepackage{hyperref}
\hypersetup{
    colorlinks=true,
    linkcolor=violet,
    filecolor=magenta,      
    urlcolor=cyan,
    citecolor=blue,
    pdffitwindow=true,
}\usepackage[capitalize, nameinlink]{cleveref}
\usepackage[backend=bibtex, style=alphabetic, backref=true,maxbibnames=99, url=false]{biblatex} 
\usepackage{algorithm}
\usepackage{algpseudocode}

\theoremstyle{plain}

\usepackage[T1]{fontenc}

\newtheorem{theorem}{Theorem}
\newtheorem{claim}[theorem]{Claim}
\newtheorem*{claim*}{Claim}
\newtheorem{lemma}[theorem]{Lemma}
\newtheorem*{lemma*}{Lemma}

\newtheorem*{prop*}{Proposition}
\newtheorem*{theorem*}{Theorem}
\newtheorem{defn}[theorem]{Definition}
\newtheorem*{defn*}{Definition}

\newtheorem*{convention*}{Convention}

\iftoggle{draft}{
\newcommand{\bedit}[1]{\Authoredit{magenta}{#1}}

}
{
\newcommand{\bedit}[1]{#1}

}

\def\OPT{\mathsf{OPT}}

\definecolor{bleudefrance}{rgb}{0.19, 0.55, 0.91}
\definecolor{blizzardblue}{rgb}{0.67, 0.9, 0.93}

\newcommand{\E}[1]{{\mathbb{E}}\left[#1\right]}
\renewcommand{\P}[1]{{\mathbb{P}}\left[#1\right]}

\usepackage{multicol}

\begin{document}

\title{A Lower Bound for the Max Entropy Algorithm for TSP}
\iftoggle{notblinded}{
\author{Billy Jin\\Purdue University \and Nathan Klein\\Boston University \and David P.\ Williamson\\Cornell University}
}
{\author{}}

% \pagenumbering{gobble}

\maketitle

\begin{abstract}{One of the most famous conjectures in combinatorial optimization is the four-thirds conjecture, which states that the integrality gap of the Subtour LP relaxation of the TSP is equal to $\frac43$. For 40 years, the best known upper bound was $1.5$, due to 
Wolsey \cite{Wolsey80}.  Recently, Karlin, Klein, and Oveis Gharan \cite{KKO21b} showed that the max entropy algorithm for the TSP gives an improved bound of $1.5 - 10^{-36}$.
In this paper, we show that the approximation ratio of the max entropy algorithm  is at least 1.375, even for graph TSP.
Thus the max entropy algorithm does not appear to be the algorithm that will ultimately resolve the four-thirds conjecture in the affirmative, should that be possible.}
\end{abstract}

% \keywords{traveling salesman problem, TSP, approximation algorithm, lower bound, graph algorithm, max entropy, Subtour LP, integrality gap}

\section{Introduction}

In the traveling salesman problem (TSP), we are given a set of $n$ cities and the costs $c_{ij}$ of traveling from city $i$ to city $j$  for all $i,j$. The goal of the problem is to find the cheapest tour that visits each city exactly once and returns to its starting point.
An instance of the TSP is called {\em symmetric} if $c_{ij} = c_{ji}$ for all $i,j$; it is {\em asymmetric} otherwise.  Costs obey the {\em triangle inequality} (or are {\em metric}) if $c_{ij} \leq c_{ik} + c_{kj}$ for all $i,j,k$.  All instances we consider will be symmetric and obey the triangle inequality. We treat the problem input as a complete graph $G=(V,E)$, where $V$ is the set of cities, and $c_e = c_{ij}$ for edge $e=(i,j)$.
% For ease of exposition, we consider the problem input as a complete graph (undirected for symmetric instances, and directed for asymmetric instances) $G=(V,E)$ for the set of cities $V$, with $c_e = c_{ij}$ for edge $e=(i,j)$.  All instances we consider will be symmetric and obey the triangle inequality.

In the mid-1970s, Christofides \cite{Christofides76} and Serdyukov \cite{Serdyukov78} each gave a
$\frac{3}{2}$-approximation algorithm for the symmetric TSP with
triangle inequality.  The algorithm computes a minimum-cost spanning tree and then finds a minimum-cost perfect matching on the odd degree vertices of the tree to compute a connected Eulerian subgraph. Because the edge costs satisfy the triangle inequality, any Eulerian tour of this Eulerian subgraph can be ``shortcut'' to a tour of no greater cost.  Until very recently, this was the best approximation factor known for the symmetric TSP with triangle inequality, although over the last decade substantial progress was made for many special cases and variants of the problem. For example, in {\em graph TSP}, the input to the problem is an unweighted connected graph, and the cost of traveling between any two nodes is the number of edges in the shortest path between the two nodes.  A sequence of papers led to a 1.4-approximation algorithm for this problem due to Seb\H{o} and Vygen \cite{SeboV14}.

In the past decade, a variation on the Christofides-Serdyukov algorithm has been considered. Its starting point is  a well-known linear programming relaxation of the TSP introduced by Dantzig, Fulkerson, and Johnson \cite{DFJ54}, sometimes called the {\em Subtour LP} or the {\em Held-Karp bound} \cite{HeldK71}.  The Subtour LP is as follows:
\begin{equation}
\begin{aligned}
    \min \quad & \sum_{e \in E} c_e x_e \\
    \text{s.t.} \quad & x(\delta(v)) = 2, &\forall \; v \in V, \\
    & x(\delta(S)) \geq 2, & \forall \; S\subset V, S \neq \emptyset, \\
    & 0 \leq x_e \leq 1, &\forall e \in E,
\end{aligned}
\label{LP}
\end{equation}
where $\delta(S)$ is the set of all edges with exactly one endpoint in $S$ and we use the shorthand that $x(F) = \sum_{e \in F} x_e$.   It is not difficult to show that for any  solution $x^*$ of this LP relaxation, $\frac{n-1}{n}x^*$ is a feasible point in the spanning tree polytope.  The spanning tree polytope is known to have integer extreme points, and so $\frac{n-1}{n}x^*$ can be decomposed into a convex combination of spanning trees, and the cost of this convex combination is a lower bound on the cost of an optimal tour.  The convex combination can be viewed as a distribution over spanning trees such that the expected cost of a spanning tree sampled from this distribution is a lower bound on the cost of an optimal tour.  The variation of the Christofides-Serdyukov algorithm considered is one that samples a random spanning tree from a distribution on spanning trees given by the convex combination and then finds a minimum-cost perfect matching on the odd vertices of the tree.  This idea was introduced in work of Asadpour, Goemans, M\k{a}dry, Oveis Gharan, and Saberi \cite{AsadpourGMOS17} (in the context of the asymmetric TSP) and Oveis Gharan, Saberi, and Singh \cite{OveisGharanSS11} (for symmetric TSP).

Asadpour et al.\ \cite{AsadpourGMOS17} and Oveis Gharan, Saberi, and Singh \cite{OveisGharanSS11} consider a particular distribution of spanning trees known as the {\em maximum entropy distribution}.  The maximum entropy algorithm finds a probability distribution $p_T$ on spanning trees $T$ such that the marginal distribution on each edge $e$ is $\frac{n-1}{n}x^*_e$ (that is, $\sum_{T:e\in T} p_T = \frac{n-1}{n} x^*_e$) and that maximizes the entropy function $-\sum_T p_T \log p_T$.    We will call the algorithm that samples from the maximum entropy distribution and then finds a minimum-cost perfect matching on the odd degree vertices of the tree the {\em maximum entropy algorithm} for the symmetric TSP.  For computational feasibility, both Asadpour et al.\ \cite{AsadpourGMOS17} and Oveis Gharan, Saberi, and Singh \cite{OveisGharanSS11} instead draw a tree from an approximation of the maximum entropy distribution by relaxing the marginal constraints so that $\sum_{T:e\in T} p_T \leq (1 + \epsilon) \frac{n-1}{n} x^*_e$, where the runtime of the algorithm depends on $\epsilon$. A sampling algorithm of Asadpour et al. \cite[Theorem 5.2]{AsadpourGMOS17} computes a random spanning tree $T$ from the maximum entropy distribution with the relaxed marginals for any point $z$ in the spanning tree polytope, not just $\frac{n-1}{n}x^*$ with $x^*$ from the Subtour LP.  
%It is known that the maximum entropy distribution finds a solution $p$ with $p_T > 0$ for any tree $T$ in the support of the LP solution $x$ when $\frac{n-1}{n}x$ is in the interior of the spanning tree polytope.  For such solutions in the interior of the spanning tree polytope, it is possible to derive values $\lambda_e$ for each edge such that the probability $p_T$ of tree $T$ is proportional to $\prod_{e \in T} \lambda_e$.  

In a breakthrough result, Karlin, Klein, and Oveis Gharan \cite{KKO21} show that a variant of the maximum entropy algorithm has performance ratio better than 3/2, although the amount by which the bound was improved is quite small (approximately $10^{-36}$). The achievement of the paper is to show that choosing a random spanning tree from the maximum entropy distribution gives a distribution of odd degree nodes in the spanning tree such that the expected cost of the perfect matching is cheaper (if marginally so) than in the Christofides-Serdyukov analysis.  Note that the Karlin et al.\ algorithm actually samples from the set of {\em 1-trees}, a spanning tree plus one additional edge\footnote{Held and Karp \cite{HeldK71} define a 1-tree to be a tree with two distinct edges incident on a specific vertex plus a spanning tree on the remaining vertices. Our definition is more general.}, and then finds a matching on the odd-degree vertices.

%Note that \cite{KKO21} actually choose a tree plus an edge, thus working with $x^*$ instead of $%\frac{n-1}{n}x^*$. Since it is cleaner to analyze we will work with this version of the algorithm. %Added this last couple of sentences, they could maybe be phrased better -Nathan

%In a breakthrough result, Karlin, Klein, and Oveis Gharan \cite{KKO21} gave the first approximation algorithm with performance ratio for the general case better than 3/2, although the amount by which the bound was improved is quite small (approximately $10^{-36}$).  The algorithm follows the Christofides-Serdyukov template by selecting a random spanning tree, then using a $T$-join on the odd degree vertices of the tree to create a connected Eulerian subgraph.  The achievement of the paper is to show that choosing a random spanning tree gives a distribution of odd degree nodes in the spanning tree such that the expected cost of the $T$-join is cheaper (if marginally so) than in the Christofides-Serdyukov analysis.

It has long been conjectured that there should be a 4/3-approximation algorithm for the TSP based on rounding the Subtour LP, given other conjectures about the integrality gap of the Subtour LP.  
 The {\em integrality gap} of an LP relaxation is the worst-case ratio of an optimal integer solution to the linear program to the optimal linear programming solution.  Wolsey \cite{Wolsey80} (and later Shmoys and Williamson \cite{ShmoysW90}) showed that the analysis of the Christofides-Serydukov algorithm could be used to show that the integrality gap of the Subtour LP is at most 3/2, and Karlin, Klein, and Oveis Gharan \cite{KKO21b} have shown that the integrality gap is at most $\frac{3}{2}-10^{-36}$. Gurvits, Klein, and Leake improved this slightly to $\frac{3}{2}-10^{-34}$ \cite{GKL24}.  It is known that the integrality gap of the Subtour LP is at least 4/3, due to a set of instances shown in Figure \ref{fig:badexample}.  It has long been conjectured
%\footnote{The first place that the authors are aware of a published statement of the conjecture is in a 1995 paper of Goemans \cite{Goemans95}, but the conjecture was in circulation earlier than that.} 
that the integrality gap is exactly 4/3. Benoit and Boyd \cite{BenoitB08} give a more refined version of the conjecture and confirm that the integrality gap is at most 4/3 when the number of vertices is at most 10; Boyd and Elliott-Magwood \cite{BoydE10} extend this result to 12 vertices.  Until the work of Karlin et al.\ there had been no progress on that conjecture for general $n$ since the work of Wolsey in 1980.    

A reasonable question is whether the maximum entropy algorithm is itself a 4/3-approximation algorithm for the TSP; there is no reason to believe that the Karlin et al. \cite{KKO21} analysis is tight.  Experimental work by Genova and Williamson \cite{GenovaW17} has shown that the max entropy algorithm produces solutions which are very good in practice, much better than those of the Christofides-Serdyukov algorithm.  It does extremely well on instances of graph TSP, routinely producing solutions within 1\% of the value of the optimal solution.

\begin{figure}
\begin{center}
\begin{tikzpicture}[inner sep=1.7pt,scale=.8,pre/.style={<-,shorten <=2pt,>=stealth,thick}, post/.style={->,shorten >=1pt,>=stealth,thick}]	
	\tikzstyle{every node} = [draw,circle,fill=black]
	\node[draw=none,fill=none] at (-1,5) (x1) {};
	\node[draw=none,fill=none] at (-1,1) (x2) {};
	\node at (0,5) (a1) {} ;
	\node at (0,4)  (a2) {};
	\node at (0,3) (a3) {} ;
	\node at (0,2)  (a4) {};
	\node at (0,1) (a5) {} ;
	\node at (1,4.5) (b1)  {};
	\node at (1,4)  (b2) {};
	\node at (1,3) (b3) {};
	\node at (1,2) (b4) {};
	\node at (1,1.5) (b5) {};
	\node at (2,5) (c1) {};
	\node at (2,4)  (c2) {};
	\node at (2,3) (c3) {};
	\node at (2,2) (c4) {};
	\node at (2,1) (c5) {};
	\path (a1) edge  (c1)
	edge   (b1)
	edge  (a2)
	(a2) edge (a3)
	(a3) edge (a4)
	(a4) edge (a5)
	(a5) edge (b5)
	edge (c5)
	(b1) edge (c1)
	edge (b2)
	(b2) edge (b3) 
	(b3) edge (b4)
	(b4) edge (b5)
	(b4) edge (b5)
	(b5) edge (c5)
	(c1) edge (c2)
	(c2) edge (c3)
	(c3) edge (c4)
	(c4) edge (c5);    
		  \path (x1) edge[<->] node[draw=none,fill=none,rectangle,left] {$k$} (x2);
\end{tikzpicture}
\qquad
\begin{tikzpicture}[inner sep=1.7pt,scale=.8,pre/.style={<-,shorten <=2pt,>=stealth,thick}, post/.style={->,shorten >=1pt,>=stealth,thick}]	
	\tikzstyle{every node} = [draw,circle,fill=black]
	\node at (0,5) (a1) {} ;
	\node at (0,4)  (a2) {};
	\node at (0,3) (a3) {} ;
	\node at (0,2)  (a4) {};
	\node at (0,1) (a5) {} ;
	\node at (1,4.5) (b1)  {};
	\node at (1,4)  (b2) {};
	\node at (1,3) (b3) {};
	\node at (1,2) (b4) {};
	\node at (1,1.5) (b5) {};
	\node at (2,5) (c1) {};
	\node at (2,4)  (c2) {};
	\node at (2,3) (c3) {};
	\node at (2,2) (c4) {};
	\node at (2,1) (c5) {};
	\path (a1) edge[dashed]  (c1)
	edge[dashed]   (b1)
	edge  (a2)
	(a2) edge (a3)
	(a3) edge (a4)
	(a4) edge (a5)
	(a5) edge[dashed] (b5)
	edge[dashed] (c5)
	(b1) edge[dashed] (c1)
	edge (b2)
	(b2) edge (b3) 
	(b3) edge (b4)
	(b4) edge (b5)
	(b4) edge (b5)
	(b5) edge[dashed] (c5)
	(c1) edge (c2)
	(c2) edge (c3)
	(c3) edge (c4)
	(c4) edge (c5);
\end{tikzpicture}
\hspace{5mm}
\begin{tikzpicture}[inner sep=1.7pt,scale=.8,pre/.style={<-,shorten <=2pt,>=stealth,thick}, post/.style={->,shorten >=1pt,>=stealth,thick}]	
	\tikzstyle{every node} = [draw,circle,fill=black]
	\node at (0,5) (a1) {} ;
	\node at (0,4)  (a2) {};
	\node at (0,3) (a3) {} ;
	\node at (0,2)  (a4) {};
	\node at (0,1) (a5) {} ;
	\node at (1,4.5) (b1)  {};
	\node at (1,4)  (b2) {};
	\node at (1,3) (b3) {};
	\node at (1,2) (b4) {};
	\node at (1,1.5) (b5) {};
	\node at (2,5) (c1) {};
	\node at (2,4)  (c2) {};
	\node at (2,3) (c3) {};
	\node at (2,2) (c4) {};
	\node at (2,1) (c5) {};
	\path (a1) edge  (a2)
	edge (b1)
	(a2) edge (a3)
	(a3) edge (a4)
	(a4) edge (a5)
	(a5) edge (b5)
	(b1) edge (b2)
	(b2) edge (b3) 
	(b3) edge (b4)
	(b5) edge (c5)
	(c1) edge (c2)
	(c2) edge (c3)
	(c3) edge (c4)
	(c4) edge (c5)
	(c1) edge (b4);
\end{tikzpicture}
\end{center}
\caption{Illustration of the worst example known for the integrality gap for
the symmetric TSP with triangle inequality.  The figure on the
left gives a graph, and costs $c_{ij}$ are the shortest path
distances in the graph.  The figure in the center gives the LP
solution, in which the dotted edges have value 1/2, and the solid
edges have value 1.  The figure on the right gives the optimal
tour.  The ratio of the cost of the optimal tour to the value of the LP solution tends to 4/3 as $k$ increases.} \label{fig:badexample}
\end{figure}

In this paper, we show that the maximum entropy algorithm can produce tours of cost strictly greater than 4/3 times the value of the optimal tour (and thus the Subtour LP), even for instances of graph TSP. In particular, we show:
\begin{theorem}
    There is an infinite family of instances of graph TSP for which the max entropy algorithm outputs a tour of expected cost at least $1.375-o(1)$ times the cost of the optimum solution. 
\end{theorem}

The instances we consider are a variation on a family of TSP instances recently introduced in the literature by Boyd and Seb\H{o} \cite{BoydS21} known as {\em $k$-donuts} (see \Cref{fig:kdonut}). 
% and characterized by a particular extreme point of the Subtour LP (see \Cref{fig:kdonut}). 
$k$-donuts have $n=4k$ vertices, and are known to have an integrality gap of 4/3 under a particular metric. In contrast, we consider $k$-donuts under the graphic metric, in which case the optimal tour is a Hamiltonian cycle, which has cost $4k+2$. We also modify one square of the $k$-donut to simplify our use of the maximum entropy algorithm.  The objective value of the Subtour LP for our graphic $k$-donut variant is also $4k+2$; thus, these instances have an integrality gap of 1. We show that as the instance size grows, the expected length of the connected Eulerian subgraph found by the max entropy algorithm (using the graphic metric) converges to $1.375n$ from below and thus the ratio of this cost to the value of the LP (and the optimal tour) converges to 1.375.  We can further show that there is a bad Eulerian tour of the Eulerian subgraph such that shortcutting the Eulerian tour results in a tour that is still at least 1.375 times the cost of the optimal tour.

The version of the maximum entropy algorithm we analyze is the one used by Karlin, Klein, and Oveis Gharan \cite{KKO20,KKO21}.  In this algorithm, they select an edge $e^+$ with $x^*_{e^+}=1$, remove this edge from the graph, and then sample a spanning tree from the remaining LP solution according to the maximum entropy distribution.  It is known that for any extreme point solution $x^*$ of the Subtour LP, there is at least one edge $e$ with $x^*_e=1$ (see, for example, Boyd and Pulleyblank \cite{BoydP91}). As before, for computational feasibility they allow a small error on the marginal constraints. That is, this algorithm finds a distribution $p_T$ on spanning trees in the graph with $e^+$ removed that maximizes the objective function  $-\sum_T p_T \log p_T$ subject to the relaxed marginal equations $\sum_{T:e\in T} p_T \le (1+\epsilon)x^*_e$ for all edges $e \neq e^+$. In this work, we show that on the $k$-donut it is possibly to efficiently sample from the max entropy distribution with no error on the marginal equations, i.e., we can set $\epsilon = 0$. Therefore we analyze this idealized version. % We also show that, as far as approximation is concerned, our variant of the maximum entropy algorithm is at most $o(1)$ different from previous approximate versions of maximum entropy algorithms considered by Oveis Gharan, Saberi, and Singh \cite{OveisGharanSS11}, Karlin, Klein, and Oveis Gharan for half-integral instances \cite{KKO20}, and Karlin, Klein, and Oveis Gharan for general instances \cite{KKO21b}.

%We show that the maximum entropy algorithm has a particularly simple form when applied to these graphic $k$-donuts: for each square of the $k$-donut with edges of LP value 1 on both the inner and outer ring, the algorithm selects both edges of LP value 1, and exactly one of the two edges of LP value 1/2, both equally likely.  Then for each square of the $k$-donut with edges of LP value 1/2 on both the inner and outer ring, the algorithm selects exactly one of the two edges from the inner and outer ring, both equally likely.  In the Appendix we prove that our algorithm is indeed the maximum entropy algorithm; that is, given a Subtour LP solution $x^*$, it finds a distribution $p_T$ on 1-trees that maximizes the objective function  $-\sum_T p_T \log p_T$ subject to the exact marginal equation $\sum_{T:e\in T} p_T = x^*_e$.  We also show that, as far as approximation is concerned, our variant of the maximum entropy algorithm is at most $o(1)$ different from previous approximate versions of maximum entropy algorithms considered by Oveis Gharan, Saberi, and Singh \cite{OveisGharanSS11}, Karlin, Klein, and Oveis Gharan for half-integral instances \cite{KKO20}, and Karlin, Klein, and Oveis Gharan for general instances \cite{KKO21b}.

It thus appears that the maximum entropy algorithm is not the algorithm that will ultimately resolve the 4/3 conjecture in the affirmative, should that be possible. While this statement depends on the fact that there is a bad Eulerian tour of the connected Eulerian subgraph, all work in this area of which we are aware considers the ratio of the cost of the connected Eulerian subgraph to the LP value, rather than the ratio of the shortcut tour to the LP value.  We also do not know of work which shows that there is always a way to shortcut the subgraph to a tour of significantly cheaper cost.  Indeed, it is known that finding the best shortcutting is an NP-hard problem in itself \cite{PapadimitriouV84}.

% \begin{tabular}{c|c|c|c|c|c|c}
%     $k$ &  \\
%     Approximation Ratio & 
% \end{tabular}

\begin{figure}[htb!]
\centering
\begin{subfigure}[t]{0.48\textwidth}
\centering
\begin{tikzpicture}[node distance=1cm, every node/.style={circle, draw, fill=black, inner sep=0pt, minimum size=5pt}, scale=0.9]
% Nodes
\foreach \i in {1,...,8} {
  \node (u\i) at (\i*45+22.5:3) {};
  \node (v\i) at (\i*45+22.5:1.5) {};
}

\node (w0) at (105:2.25) {};
\node (w1) at (75:2.25) {};
% Matching edges 
\foreach \i [evaluate=\i as \j using {int(Mod(\i, 8)+1)}] in {2,4,6,8} {
  \draw (u\i) -- (u\j);
  \draw (v\i) -- (v\j);
}

% \node (u_head) at (67.5:3) [label=right:{$u_{0}$}] {};
% \node (u_tail) at (112.5:3) [label=left:{$u_{2k-1}$}] {};
% \node (v_head) at (67.5:1.5) [label=right:{$v_{0}$}] {};
% \node (v_tail) at (112.5:1.5) [label=left:{$v_{2k-1}$}] {};

\draw (w0) -- (w1);

% Square edges 
\foreach \i [evaluate=\i as \j using {int(Mod(\i, 8)+1)}] in {1,3,5,7} {
  \draw (u\i) -- (u\j);
  \draw (v\i) -- (v\j);
}

\draw (w0) -- (u2);
\draw (w0) -- (v2);
\draw (w1) -- (u1);
\draw (w1) -- (v1);

% Connecting edges
\foreach \i in {1,...,8} {
  \draw (u\i) -- (v\i);
}
\end{tikzpicture}
\caption{Our variant on the $k$-donut for $k=4$. There are $n=4k+2$ vertices. All edges shown have cost 1, and all edges not shown have cost equal to the shortest path distance. We will refer to the outer cycle as the \emph{outer ring}, and the inner cycle as the \emph{inner ring}.
}
\label{fig:kdonut_base}
\end{subfigure}
\hfill
\begin{subfigure}[t]{0.48\textwidth}
    \centering
    \begin{tikzpicture}[node distance=1cm, every node/.style={circle, draw, fill=black, inner sep=0pt, minimum size=5pt},scale=0.9]
% Nodes
\foreach \i in {1,...,8} {
  \node (u\i) at (\i*45+22.5:3) {};
  \node (v\i) at (\i*45+22.5:1.5) {};
}

\node (w0) at (105:2.25) {};
\node (w1) at (75:2.25) {};
% Matching edges 
\foreach \i [evaluate=\i as \j using {int(Mod(\i, 8)+1)}] in {2,4,6,8} {
  \draw (u\i) -- (u\j);
  \draw (v\i) -- (v\j);
}

% \node (u_head) at (67.5:3) [label=right:{$u_{0}$}] {};
% \node (u_tail) at (112.5:3) [label=left:{$u_{2k-1}$}] {};
% \node (v_head) at (67.5:1.5) [label=right:{$v_{0}$}] {};
% \node (v_tail) at (112.5:1.5) [label=left:{$v_{2k-1}$}] {};

\draw (w0) -- (w1);

% Square edges 
\foreach \i [evaluate=\i as \j using {int(Mod(\i, 8)+1)}] in {3,5,7} {
  \draw[dash pattern=on 2pt off 2pt] (u\i) -- (u\j);
  \draw[dash pattern=on 2pt off 2pt] (v\i) -- (v\j);
}

\draw[dash pattern=on 2pt off 2pt] (w0) -- (u2);
\draw[dash pattern=on 2pt off 2pt] (w0) -- (v2);
\draw[dash pattern=on 2pt off 2pt] (w1) -- (u1);
\draw[dash pattern=on 2pt off 2pt] (w1) -- (v1);

% Connecting edges
\foreach \i in {1,...,8} {
  \draw[dash pattern=on 2pt off 2pt] (u\i) -- (v\i);
}
\end{tikzpicture}
\caption{An extreme point optimal solution $x$ to the Subtour LP. The dotted edges have $x_e = \frac12$ and the solid edges have $x_e =1$.
}
\label{fig:kdonut_lp}
\end{subfigure}
\caption{The $k$-donut instance.}
% \caption{Our variant on the $k$-donut for $k=4$, where $k$ indicates the number of squares of dotted edges. There are $n=4k+2$ vertices. The dotted edges have $x_e = \frac12$ and the solid edges have $x_e =1$ in the LP solution. The solid and dotted edges all have cost 1. The edges $u_0u_{2k-1}$ and $v_0v_{2k-1}$ also have cost 1. All other edges have their cost equal to the shortest path distance. 
% We will refer to the outer cycle as the \emph{outer ring}, and the inner cycle as the \emph{inner ring}.
% }
\label{fig:kdonut}
\end{figure}

For our instances, $\frac{4}{3}$-approximation algorithms are already known.  For example, M\"omke and Svensson \cite{MomkeS16} have shown how to obtain a $\frac{4}{3}$-approximation algorithm for half-integral instances of graph TSP.  Earlier work of the authors \cite{JinKW23} gives a $\frac{4}{3}$-approximation for half-integral cycle cut instances of the TSP, a class which includes this $k$-donut variant.

%Interestingly, earlier work of the authors \cite{JinKW23} gave a 4/3-approximation algorithm for half-integral cycle cut instances of the TSP, a class which includes these $k$-donuts. %Therefore this is some evidence that max entropy may not be the best possible approximation algorithm for TSP.
%The two works together demonstrate that max entropy is not the optimal polynomial time algorithm for all instances of metric TSP. %Added this last sentence, not sure if it's needed -Nathan

Our work continues a thread of papers showing lower bounds on TSP approximation algorithms.  Rosenkrantz, Stearns, and Lewis \cite{RosenkrantzSL77} show that the nearest neighbor algorithm can give a tour of cost $\Omega(\log n)$ times the optimal, and that the nearest and cheapest insertion algorithms can give a tour of cost $2 - \frac{1}{n}$ times the optimal.  Cornu\'{e}jols and Nemhauser \cite{CornuejolsN78} show that the Christofides-Serdyukov approximation factor of $\frac{3}{2}$ is essentially tight.  Chandra, Karloff, and Tovey \cite{ChandraKT06} show that the 2-OPT local search heuristic can return a tour of cost $\Omega(\sqrt{n})$ of the optimal cost and give a bound of $\Omega(\sqrt[2k]{n})$ for $k$-OPT heuristic; Zhong \cite{Zhong21} has improved some of these bounds.

\section{$k$-Donuts}

We first formally describe the construction of a graphic $k$-donut instance, which will consist of $4k+2$ vertices. The cost function $c_{\{u,v\}}$ is given by the shortest path distance in the following graph. 

\begin{defn}[$k$-Donut Graph]
For $k \in \mathbb{Z}_+$ with $k \ge 3$, the \emph{$k$-donut} is the graph consisting of $2k$ \emph{outer} vertices $u_0, \dots, u_{2k-1}$, $2k$ \emph{inner} vertices $v_0, \dots, v_{2k-1}$, and two special vertices $w_0$ and $w_1$. For each $0 \le i \le 2k - 1$, the graph includes the edges $\{u_i, u_{(i+1) \bmod 2k}\}$, $\{v_i, v_{(i+1) \bmod 2k}\}$, and $\{u_i, v_i\}$. In addition, the graph contains the edges $\{w_0, w_1\}$, $\{w_0, u_{0}\}$, $\{w_0, v_{0}\}$, $\{w_1, u_1\}$, and $\{w_1, v_1\}$. See \Cref{fig:kdonut_base} for an illustration. We refer to the cycle formed by the $u$-vertices as the \emph{outer ring}, and the cycle formed by the $v$-vertices as the \emph{inner ring}.
\end{defn}

For clarity of notation, we will omit the "mod $2k$" operation when indexing the vertices of the $k$-donut throughout the remainder of the paper. Thus, whenever we write $u_j$ or $v_j$, it should be understood as $u_{j \bmod 2k}$ or $v_{j \bmod 2k}$, respectively. \bedit{We also refer to the subgraph induced by the six vertices $\{u_0,v_0,w_0,w_1,u_1,v_1 \}$ as the "envelope gadget".}

In the graphic $k$-donut instance, there exists a {half-integral} extreme point solution $x$ to the Subtour LP with value $4k + 2$, which we will use throughout the paper. \begin{defn}[$k$-Donut Extreme Point]\label{def:extremepoint}
    Let $x_{\{u_i, v_i\}} = 1/2$ for all $0 \le i \le 2k - 1$; let $x_{\{u_i, u_{i+1}\}} = x_{\{v_i, v_{i+1}\}} = 1$ for all odd $i$; and let $x_{\{u_i, u_{i+1}\}} = x_{\{v_i, v_{i+1}\}} = 1/2$ for all even $i \ne 0$. Additionally, set $x_{\{w_0, w_1\}} = 1$ and $x_{\{w_0, u_{0}\}} = x_{\{w_0, v_{0}\}} = x_{\{w_1, u_1\}} = x_{\{w_1, v_1\}} = 1/2$.\footnote{By slightly perturbing the metric, one can ensure that $x$ is the \textit{unique} optimal solution to the LP, and thus the solution used by the max-entropy algorithm. (Of course, the instance would then no longer be strictly graphic.)}
\end{defn}   
See \Cref{fig:kdonut_lp} for an illustration of $x$.

\begin{restatable}{prop}{extremepoint}
    The point $x$ given in \cref{def:extremepoint} is an extreme point solution to the Subtour LP. 
\end{restatable}

We prove this in \cref{sec:appendix}. In the rest of the paper, we will say that a set $S \subseteq V$ is \emph{tight} if $x(\delta(S)) = 2$, and \emph{proper} if $2 \le |S| \le |V| - 2$. For a set of edges $M$, we define $c(M) = \sum_{e \in M} c_e$. For an LP solution $x$, we let $c(x)$ denote the value of the LP objective function, i.e., $c(x) = \sum_{e \in E} c_e x_e$. For a subset $S \subseteq V$, we use $E(S)$ to denote the set of edges with both endpoints in $S$.

Our $k$-donut construction and associated LP solution are inspired by a similar instance of Boyd and Seb\H{o} \cite{BoydS21}. 

\subsection{The Max Entropy Algorithm on the $k$-Donut}

We now describe the max entropy algorithm on general graphs, and its behavior when specialized to the $k$-donut. We work with the version of the max entropy algorithm used in \cite{KKO21} and \cite{KKO20}. In both works, the authors show that, without loss of generality, there exists an edge $e^+$ with $x_{e^+} = 1$.\footnote{The argument is as follows. If such an edge does not exist,  split a node $v$ into two nodes $v_1$, $v_2$; connect 2 of the edges out of $v$ to $v_1$ and the other two to $v_2$. Then, connect $v_1$ to $v_2$ with an edge $e^+$ of cost
$c(e^+) = 0$ and $x_{e^+} = 1$.} In our instances, such an edge exists by construction; we will use $e^+=\{w_0,w_1\}.$

To sample a 1-tree $T$, the algorithm begins by adding $e^+$ to $T$. It then samples a spanning tree from the max entropy distribution with marginals given by the LP solution with $e^+$ removed, and adds this tree to $T$.  In other words, let $\mathcal{X}$ be the set of spanning trees in the support of $x$ with $e^+$ removed. The algorithm solves for the distribution $\{p_X : X \in \mathcal{X}\}$ that maximizes the entropy objective $-\sum_{X \in \mathcal{X}} p_X \log p_X$, subject to the exact marginal constraints $\sum_{X : e \in X} p_X = x_e$ for all edges $e \ne e^+$. A spanning tree $X$ is then sampled from this distribution and added to $T$. Finally, the algorithm computes a minimum-cost perfect matching $M$ on the odd-degree vertices of $T$, constructs an Eulerian tour on $M \uplus T$, and shortcuts the tour to obtain a Hamiltonian cycle.\footnote{Given an Eulerian tour $(t_0,\dots,t_\ell)$, we shortcut it to a Hamiltonian cycle by keeping only the first occurrence of every vertex except $t_0$. Due to the triangle inequality, the resulting Hamiltonian cycle has cost no greater than that of the Eulerian tour.}

The above describes the maximum entropy algorithm for general graphs. Algorithm~\ref{alg} presents its specialization to the $k$-donut instance. In \Cref{lem:max_entropy_kdonut} in the Appendix, we prove that Algorithm~\ref{alg} is indeed the correct specialization of the maximum entropy algorithm to the $k$-donut instance. See \Cref{fig:kdonut_tree} for an illustration of how the algorithm selects edges from the $k$-donut to put into the 1-tree.

\begin{algorithm}[h]
\caption{Max Entropy Algorithm on the $k$-Donut}

{\small{\textbf{Note.} All sampling in this algorithm is done independently and uniformly at random.
}}

\begin{algorithmic}[1]
\State Set $T = \emptyset$.\Comment{$T$ will be a 1-tree} 
\State Let $e^+ = \{w_0, w_1\}$; add $e^+$ to $T$. 
\State For each odd $i$, add the edges $\{u_i, u_{i+1}\}$ and $\{v_i, v_{i+1}\}$ to $T$.
\State For each odd $i$, sample one edge from the pair $\{\{u_i, v_i\}, \{u_{i+1}, v_{i+1}\}\}$ and add it to $T$.
\State For each even $i$, sample one edge from the pair $\{\{u_i, u_{i+1}\}, \{v_i, v_{i+1}\}\}$ and add it to $T$.
\State Sample one edge from the pair $\{\{w_0, u_0\}, \{w_0, v_0\}\}$ and add it to $T$.
\State Sample one edge from the pair $\{\{w_1, u_1\}, \{w_1, v_1\}\}$ and add it to $T$.
\State Compute the minimum-cost perfect matching $M$ on the odd vertices of $T$. Compute an Eulerian tour of $T \uplus M$ and shortcut it to return a Hamiltonian cycle.
\end{algorithmic}
\label{alg}
\end{algorithm}

\begin{figure}[htb!]
\centering
\begin{tikzpicture}[node distance=2cm, minimum size=2pt, every node/.style={circle, draw,scale=0.6},font=\tiny]

\def\offset{0.5} % change this value as desired

\foreach \i in {1,2,3,4} {
  % Determine horizontal shift direction
    \node (u\i) at ($(-\i*45+112.5:3) + (\offset, 0)$) {$u_{\i}$};
    \node (v\i) at ($(-\i*45+112.5:1.5) + (\offset, 0)$) {$v_{\i}$};
}
\foreach \i in {0,5,6,7} {
    \node (u\i) at ($(-\i*45+112.5:3) - (\offset, 0)$) {$u_{\i}$};
    \node (v\i) at ($(-\i*45+112.5:1.5) - (\offset, 0)$) {$v_{\i}$};
}

% Nodes
% \foreach \i in {0,...,7} {
%   \node[font=\tiny] (u\i) at (-\i*45+112.5:3) {$u_{\i}$};
%   \node[font=\tiny] (v\i) at (-\i*45+112.5:1.5) {$v_{\i}$};
% }

% \foreach \i in {1,...,8} {
%   \node (u\i) at (\i*45+22.5:3) {
%  $u_{\i}$};
%   \node (v\i) at (\i*45+22.5:1.5) {
%  $v_{\i}$};
% }

\node[font=\tiny] (w0) at (100:2.25) {$w_0$};
\node[font=\tiny] (w1) at (80:2.25) {$w_1$};
% Matching edges 
\foreach \i [evaluate=\i as \j using {int(Mod(\i+1, 8))}] in {1,3,5,7} {
  \draw (u\i) -- (u\j);
  \draw (v\i) -- (v\j);
}

% \node (u_head) at (67.5:3) [label=right:{$u_{0}$}] {};
% \node (u_tail) at (112.5:3) [label=left:{$u_{2k-1}$}] {};
% \node (v_head) at (67.5:1.5) [label=right:{$v_{0}$}] {};
% \node (v_tail) at (112.5:1.5) [label=left:{$v_{2k-1}$}] {};

\draw (w0) -- (w1);

% Square edges 
\foreach \i [evaluate=\i as \j using {int(Mod(\i+1, 8))}] in {2,4,6} {
  \draw[dash pattern=on 2pt off 2pt] (u\i) -- (u\j);
  \draw[dash pattern=on 2pt off 2pt] (v\i) -- (v\j);
}

\draw[dash pattern=on 2pt off 2pt] (w0) -- (u0);
\draw[dash pattern=on 2pt off 2pt] (w0) -- (v0);
\draw[dash pattern=on 2pt off 2pt] (w1) -- (u1);
\draw[dash pattern=on 2pt off 2pt] (w1) -- (v1);

% Connecting edges
\foreach \i in {0,...,7} {
  \draw[dash pattern=on 2pt off 2pt] (u\i) -- (v\i);
}

\foreach \i in {1,3} {
  \draw[red] ($(-\i*45+67+22.5:2.4) + (\offset, 0)$) circle (1.5cm);
}
\foreach \i in {5,7} {
  \draw[red] ($(-\i*45+67+22.5:2.4) - (\offset, 0)$) circle (1.5cm);
}

% \foreach \i in {7} {
%   \draw[red] (-\i*45+67+22.5:2.4) ellipse [x radius=1.8cm, y radius=0.9cm, rotate=135];
% }

% \foreach \i in {1} {
%   \draw[red] (-\i*45+67+22.5:3) ellipse [x radius=2.4cm, y radius=1.2cm, rotate=45];
% }
\end{tikzpicture}
\caption{The max entropy algorithm begins by putting all of the 1-edges into the 1-tree. Then, one edge among the pair of dotted edges inside each red circled cut will be chosen independently. Next, one edge among each pair of dotted edges in the cycle resulting from contracting the red sets will be chosen independently. Finally, one of the two dotted edges incident to $w_0$ will be chosen independently, and the same for $w_1$.}\label{fig:kdonut_tree}
\end{figure}

The following claim is the only property we need in the remainder of the proof:
\begin{claim}\label{claim:1odd1even}
    For every pair of vertices $(u_i,v_i)$, $0 \le i \le 2k-1$, exactly one of $u_i$ or $v_i$ will have odd degree in $T$, each with probability $\frac12$. Let $O_i$ be the indicator random variable that $u_i$ is odd. Then, $O_0, \ldots, O_{2k-1}$ are mutually independent. 
\end{claim}
\begin{proof}
    We prove the first part of the claim when $i$ is odd; the case where $i$ is even is similar. Since $i$ is odd, the edges $\{u_i, u_{i+1}\}$ and $\{v_i, v_{i+1}\}$ are in $T$. Then, one of the two edges $\{u_i, v_i\}$ and $\{u_{i+1}, v_{i+1}\}$ is added to $T$, and regardless of the choice, $u_i$ and $v_i$ so far have the same parity. Finally, if $i > 1$, one edge in $\{\{u_{i-1}, u_i\}, \{v_{i-1}, v_i\}\}$ is added uniformly at random, and if $i = 1$, one edge in $\{\{w_1, u_i\}, \{w_{1}, v_i\}\}$ is added uniformly at random. In either case, this flips the parity of exactly one of $u_i,v_i$. 

    To show that $O_0, \ldots, O_{2k-1}$ are mutually independent, it suffices to show that for all binary vectors $b \in \{0,1\}^{2k}$, we have
    $$\mathbb{P}(O_0 = b_0, \ldots, O_{2k-1} = b_{2k-1}) = 2^{-2k}.$$
    \Cref{fig:indep} shows the edges of the $k$-donut that are incident to $u_0, \ldots, u_{2k-1}$. The presence or absence of these edges in $T$ determines the realizations of $O_0, \ldots, O_{2k-1}$.

        \begin{figure}
        \centering
        \begin{tikzpicture}[every node/.style={circle, draw, fill=black, inner sep=0pt, minimum size=5pt}]
        
        \foreach \i in {1,2,...,8} {
            \node (u\i) at (\i, 1) {};
            \node (v\i) at (\i, 0) {};
        }

         \node (u1_label) at (1, 1) [label=above:$u_1$] {};
         \node (u2_label) at (2, 1) [label=above:$u_2$] {};
         \node (u7_label) at (7, 1) [label={[yshift=-6.7pt]above:$u_{2k-1}$}] {};
         \node (u0_label) at (8, 1) [label=above:$u_0$] {};
         \node (v1_label) at (1, 0) [label=below:$v_1$] {};
         \node (v2_label) at (2, 0) [label=below:$v_2$] {};
         \node (v7_label) at (7, 0) [label={[yshift=6.7pt]below:$v_{2k-1}$}] {};
         \node (v0_label) at (8, 0) 
         [label=below:$v_0$] {};

         \node (w0) at (9, 1) [label=above:$w_0$] {};
         \node (w1) at (0, 1)[label=above:$w_1$] {};

        \draw[dash pattern=on 2pt off 2pt] (u8) -- (w0);
        \draw[dash pattern=on 2pt off 2pt] (u1) -- (w1);

        \foreach \i [evaluate=\i as \j using {\i+1}] in {1,3,5,7} {
            \draw (u\i) -- (u\j);
            \draw[dash pattern=on 2pt off 2pt]  (u\i.center) -- (v\i.center);
        }
        \foreach \j in {2,4,6,8} {
            \draw[dash pattern=on 2pt off 2pt]  (u\j) -- (v\j);
        }

        \foreach \i [evaluate=\i as \j using {\i+1}] in {2,4,6} {
            \draw[dash pattern=on 2pt off 2pt]  (u\i) -- (u\j);
        }
            
        \end{tikzpicture}
        \caption{The edges of the $k$-donut incident to $u_0, u_1, \ldots, u_{2k-1}$.}
        \label{fig:indep}
    \end{figure}

    Fix any binary vector $b \in \{0,1\}^{2k}$. First, condition on the event that $u_1w_1 \in T$. Consider $u_1$. Among the two edges $\{u_1v_1, u_2v_2\}$, one of them is selected to put into $T$, with probability $\frac12 $ each. Exactly one of these choices makes $O_1 = b_1$. Next, further condition on event $O_1 = b_1$ and consider $u_2$. The edge $u_2u_3$ is either in $T$ or not, with probability $\frac12$ each. Exactly one of these choices makes $O_2 = b_2$. Repeating this line of reasoning, we see that for each $j$, conditioned on $u_1w_1 \in T$ and $O_i = b_i$ for $i < j$, there is the probability that $O_j = b_j$ is exactly $\frac12$. Therefore, 
    $$\mathbb{P}(O_0 = b_0, \ldots, O_{2k-1} = b_{2k-1} \mid u_1w_1 \in T) = 2^{-2k}.$$
    The same argument shows that 
    $$\mathbb{P}(O_0 = b_0, \ldots, O_{2k-1} = b_{2k-1} \mid u_1w_1 \not\in T) = 2^{-2k},$$
    which implies $\mathbb{P}(O_0 = b_0, \ldots, O_{2k-1} = b_{2k-1}) = 2^{-2k}$ and completes the proof. 

\end{proof}

%The max entropy algorithm first solves the Subtour LP (\Cref{LP}) to obtain a solution $x$. It then finds a distribution $\mu$ over 1-trees (trees plus an edge) of maximum entropy such that $\mathbb{P}_{T \sim \mu}[e \in T] = x_e$ for all $e \in E$.\footnote{Note we work with the idealized version of max entropy in which all these equalities are met exactly, whereas generally some exponentially small error must be tolerated.}

\section{Analyzing the Performance of Max Entropy}
\label{sec:analysis}

We now analyze the performance of the max entropy algorithm on graphic $k$-donuts. We first characterize the structure of the min-cost perfect matching on the odd vertices of $T$. We then use this structure to show that in the limit as $k \to \infty$, the approximation ratio of the max entropy algorithm approaches $1.375$ from below.
\begin{claim}
\label{claim:M1orM2}

    Let $T$ be any 1-tree with the property that for every pair of vertices $(u_i,v_i)$ for $0 \le i \le 2k-1$, exactly one of $u_i$ or $v_i$ has odd degree in $T$. (This is \Cref{claim:1odd1even}). 

    Let $o_0,\dots,o_{2k-1}$ indicate the odd vertices in $T$ where $o_i$ is the odd vertex in the pair $(u_i,v_i)$. Let $M$ be a minimum-cost perfect matching on the odd vertices of $T$. Define: 
    $$M_1 = \{(o_0,o_1),(o_2,o_3),\dots,(o_{2k-2},o_{2k-1})\}$$
    $$M_2 = \{(o_{2k-1},o_0),(o_1,o_2),\dots,(o_{2k-3},o_{2k-2})\}$$
    Then, 
    $$c(M) = \min\{c(M_1),c(M_2)\}.$$
\end{claim}
\begin{proof}
    We will show a transformation from $M$ to a matching in which every odd vertex $o_i$ is either matched to $o_{i-1 \pmod{2k}}$ or $o_{i+1 \pmod{2k}}$. This completes the proof, since then after fixing $(o_0,o_1)$ or $(o_{2k-1},o_0)$ the rest of the matching is uniquely determined as $M_1$ or $M_2$. During the process, we will ensure the cost of the matching never increases, and to ensure it terminates we will argue that the (non-negative) potential function $\sum_{e=(o_i,o_j) \in M} \min\{|i-j|,2k-|i-j|\}$ decreases at every step. Note that this potential function is invariant under any  reindexing corresponding to a cyclic shift of the indices. Furthermore, observe that the metric on the odd vertices is unchanged for any cyclic shift of the indices (to see this, recall that the $k$-donut has edges $\{u_0,u_1\}$ and $\{v_0,v_1\}$). 

   So, suppose $M$ is not yet equal to $M_1$ or $M_2$. Then there is some edge $(o_i, o_j) \in M$ such that $j\not\in \{i-1,i+1 \pmod{2k}\}$. Without loss of generality (by switching the role of $i$ and $j$ if necessary), suppose $j \in \{i+2,i+3,\ldots,i+k \pmod{2k}\}$. Possibly after a cyclic shift of the indices, we can further assume $i = 0$ and $2 \leq j \leq k$. Let $o_l$ be the vertex that $o_1$ is matched to. We consider two cases depending on if $l \leq k+1$ or $l > k+1$.

   \textbf{Case 1:} $l \leq k+1$. In this case, replace the edges $\{\{o_0, o_j\}, \{o_1, o_l\}\}$ with  $\{\{o_0, o_1\}, \{o_j, o_l\}\}$. This decreases the potential function, as the  edges previously contributed $j+l-1$ and now contribute $1+|j-l|$, which is a smaller quantity since $j,l \geq 2$. Moreover this does not increase the cost of the matching: We have $c_{\{o_0, o_1\}} \leq 2$ and $c_{\{o_l, o_j\}} \leq |j-l|+1$, so the two new edges cost at most $|j-l|+3$. On the other hand, the two old edges cost at least $c_{\{o_0, o_j\}} + c_{\{o_1, o_l\}}\geq j+l-1$, which is at least $|j-l|+3$ since $j,l \geq 2$. 

   \textbf{Case 2:} $l > k+1$. In this case, we replace the edges $\{\{o_0, o_j\}, \{o_1, o_l\}\}$ with  $\{\{o_0, o_l\}, \{o_1, o_j\}\}$. This decreases the potential function, as the edges previously contributed $j + (2k-l+1)$ and now they contribute $(2k-l)+(j-1)$. Also, the edges previously cost at least $j + (2k-l+1)$, and now cost at most $(2k-l+1)+j$. Thus the cost of the matching did not increase.
\end{proof}

We now analyze the approximation ratio of the max entropy algorithm without shortcutting.
\begin{lemma}
If $A = T \uplus M$ is the connected Eulerian subgraph computed by the max entropy algorithm on the $k$-donut, then
    $$\lim_{k \to \infty} \frac{\E{c(A)}}{c(\OPT)} = \lim_{k \to \infty} \frac{\E{c(A)}}{c(x)} = 1.375,$$
    where $c(x)$ is the cost of the extreme point solution to the Subtour LP.
\end{lemma}
\begin{proof}
    By construction, $c(x) = 4k+2$. Since the $k$-donut is Hamiltonian, we also have that the optimal tour has length $4k+2$. The cost of the Eulerian subgraph is $c(A) = c(T) + c(M)$, where $T$ is the 1-tree and $M$ is the matching. Note that the cost of the 1-tree is always $4k+2$. Finally, we know that $c(M) = \min\{c(M_1), c(M_2)\}$ from the previous claim.
    Thus, it suffices to reason about the cost of $M_1$ and $M_2$. We know that for every $i$, $c_{\{o_i,o_{i+1 \pmod{2k}}\}}=2$ with probability 1/2 and 1 otherwise, using \Cref{claim:1odd1even}. Thus, the expected cost of each edge in $M_1$ and $M_2$ is 1.5. Since each matching consists of $k$ edges, by linearity of expectation, $\E{c(M_1)} = \E{c(M_2)} = 1.5k$. By Jensen's inequality (since $\mathrm{min}$ is a concave function), this implies $\E{c(M)} \leq 1.5k$. This immediately gives an upper bound on the approximation ratio of $\frac{4k+2+1.5k}{4k+2} = 1.375-o(1)$. In the remainder we prove the lower bound.

    For each $i$, construct a random variable $X_i$ indicating if $c_{\{o_i,o_{i+1 \pmod{2k}}\}} = 2$. By \Cref{claim:1odd1even}, the variables $\{X_0, X_1, \ldots\ X_{2k-1}\}$ are mutually independent. To analyze the cost of $M_1$, we define $\mu = \mathbb{E}[\sum_{i=0}^{k-1} X_{2i}] = k/2.$ Then we have
    \begin{align*}
        \P{c(M_1) \ge \left(\frac{3}{2}-\epsilon\right)k} &= \P{\sum_{i=0}^{k-1} X_{2i} \ge \left(\frac{1}{2}-\epsilon\right)k} \\
        &\ge 1-\P{\sum_{i=0}^{k-1} X_{2i} \le (1-2\epsilon) \mu} \\
        &\ge 1-e^{-2\epsilon^2k/3}.
    \end{align*}
    In the last line we used the mutual independence of the $X_0,\dots,X_{2k-1}$ to apply a Chernoff bound. Now, choosing $\epsilon = \sqrt{\frac{\ln(k)}{k}}$ and applying a union bound (the same bound applies to $M_2$), we obtain that the probability that both matchings cost at least $(\frac{3}{2}-o(1))k$ is at least $1-o(1)$, where $o(1)$ is a quantity that goes to 0 as $k \to\infty$. The expected cost of the matching is therefore at least $$(1-o(1))\left(\frac{3}{2}-o(1)\right)k = \left(\frac{3}{2}-o(1)\right)k.$$ Since the cost of the 1-tree is always $4k+2$, we obtain an expected cost of $(\frac{11}{2}-o(1))k$ with an approximation ratio of
    $$\frac{\E{c(T \uplus M)}}{OPT} = \frac{(\frac{11}{2}-o(1))k}{OPT} = \frac{(\frac{11}{2}-o(1))k}{4k+2} = \frac{11}{8}-o(1),$$
    which goes to $\frac{11}{8}$ as $k \to \infty$.
    % For each $i$, construct a random variable $X_i$ indicating if $c_{\{o_i,o_{i+1 \pmod{2k}}\}} = 2$. By \Cref{claim:1odd1even}, these variables are fully independent.
    % Therefore, applying a Chernoff bound for $M_1$,
    % $$\P{c(M_1) \le \left(\frac{3}{2}-\epsilon\right)k} = \P{\sum_{i=0}^{2k-1} X_{2i} \le \left(\frac{1}{2}-\epsilon\right)k} \le e^{-\epsilon^2k/4}$$
    % Choosing $\epsilon = 2k^{-1/2}\ln(k)$, the upper bound is $1/k$. So applying a union bound (the same bound applies to $M_2$), we obtain the chance that both matchings cost at least $\frac{3}{2}k-2k^{1/2}\log(k)$ occurs with probability at least $1-\frac{2}{k}$. Even if the matching has cost $k$ on the remaining instances (which is the smallest possible since there are $2k$ odd vertices), the expected cost of the matching is therefore at least 
    % $$\left(1-\frac{2}{k}\right)\left(\frac{3}{2}k-2k^{1/2}\log(k)\right) + \frac{2}{k} \cdot k \ge \frac{3}{2}k - 3k^{1/2}\log(k),$$
    % where we use $k \ge 4$. 
    % Since the cost of the 1-tree is always $4k$, we obtain an expected cost of $\frac{11}{2}k - 3k^{1/2}\log(k)$ with a ratio of
    % $$\E{c(T \uplus M)} = \frac{\frac{11}{2}k - 3k^{1/2}\log(k)}{\OPT} = \frac{\frac{11}{2}k - 3k^{1/2}\log(k)}{4k} \ge 1.375 - k^{-1/2}\log(k),$$
    % which of course goes to $1.375$ as $k \to \infty$.
\end{proof}

\section{Shortcutting}
\label{sec:shortcut}

So far, we have shown that the expected cost of the connected Eulerian subgraph returned by the max entropy algorithm is 1.375 times that of the optimal tour. However, after shortcutting the Eulerian subgraph to a Hamiltonian cycle, its cost may decrease. Ideally, we would like a lower bound on the cost of the tour after shortcutting. One challenge with this is that the same Eulerian subgraph can be shortcut to different Hamiltonian cycles with different costs, depending on which Eulerian tour is used for the shortcutting. What we will show in this section is that there is always \emph{some} bad Eulerian tour of the connected Eulerian subgraph, whose cost does not go down after shortcutting. 
% such that even after shortcutting the Eulerian tour to a Hamiltonian cycle, the expected cost stays the same. 
We highlight two important aspects of this analysis:
\begin{enumerate}
    \item In the analysis in \Cref{sec:analysis}, we did not require  the chosen matching to be either $M_1$ or $M_2$. Thus, we lower bounded the cost of the Eulerian subgraph for any procedure that obtains a minimum-cost matching (or $T$-join). Here we will require that the matching algorithm always selects $M_1$ or $M_2$. From \Cref{claim:M1orM2}, we know one of these matchings is a candidate for the minimum-cost matching. However, there may be others. Therefore, we only lower bound the shortcutting for a specific choice of the minimum-cost matching. 
    % Therefore we only lower bound the shortcutting for matching algorithms that do not specify how to tie break. (And indeed our proof requires adversarial access to the tiebreaking procedure).
\item 
% We also require that the Eulerian tour which is shortcut is chosen adversarially. Thus, 
Similar to above, we only lower bound the shortcutting for a specific Eulerian tour of the Eulerian subgraph.  Indeed, our lower bound only holds for a small fraction of Eulerian tours. 
% we require adversarial access to the algorithm that produces an Eulerian tour. 
\end{enumerate}
We remark that the max entropy algorithm as described in e.g. \cite{OveisGharanSS11,KKO21} does not specify the manner in which a minimum-cost matching or an Eulerian tour is generated. Therefore, our lower bound holds for the general description 
% a valid implementation 
of the algorithm. Thus, despite the caveats, this section successfully demonstrates our main result:
% without changing its description of 
The max entropy algorithm is not a 4/3-approximation algorithm. 

In the rest of this section, we will first describe some properties of the 1-tree that will be useful in the analysis. Then, we consider Eulerian graphs resulting from adding $M_1$ and construct Eulerian tours whose costs do not decrease after shortcutting. We then do the same for the graphs resulting from adding $M_2$. These two statements together complete the proof. 

\subsection{The Structure of the Tree}

The 1-tree $T$ will take all edges of the form $(u_i,u_{i+1})$ and all edges $(v_i,v_{i+1})$ for all odd $i$, since these edges $e$ have $x_e = 1$. 
Thus, we begin by taking every other edge in the outer ring and doing the same in the inner ring.
% Thus we begin with two alternating cycles of edges in the tree and out of the tree. 
We then proceed to take exactly one of the edges $\{\{u_i,v_i\}, \{ u_{i+1}, v_{i+1}\}\}$ for every odd $i$. Thus every block $i$ of four vertices $u_i,v_i,u_{i+1},v_{i+1}$ for odd $i$ now has three edges, and there are exactly two possibilities. We call block $i$ a 0-block if we picked edge $\{u_i,v_i\}$ and a 1-block otherwise, as seen in \Cref{fig:0-1-block}.

\begin{figure}[htb!]
    \centering 
    \begin{tikzpicture}
    % Vertices
    \node[fill=black, circle, inner sep=1.5pt, label={[label distance=-0.15cm, below]:$v_{i}$}] (A) at (0,0) {};
    \node[fill=black, circle, inner sep=1.5pt, label={[label distance=-0.15cm, below]:$v_{i+1}$}] (B) at (1,0) {};
    \node[fill=black, circle, inner sep=1.5pt, label={[label distance=-0.05cm, above]:$u_{i+1}$}] (C) at (1,1) {};
    \node[fill=black, circle, inner sep=1.5pt, label={[label distance=-0.05cm, above]:$u_{i}$}] (D) at (0,1) {};

    % Edges
    %\draw (A) -- (B) -- (C) -- (D);
    \draw (C) -- (D) -- (A) -- (B);
\end{tikzpicture}\hspace*{25mm}
    \begin{tikzpicture}
    % Vertices
     \node[fill=black, circle, inner sep=1.5pt, label={[label distance=-0.15cm, below]:$v_{i}$}] (A) at (0,0) {};
    \node[fill=black, circle, inner sep=1.5pt, label={[label distance=-0.15cm, below]:$v_{i+1}$}] (B) at (1,0) {};
    \node[fill=black, circle, inner sep=1.5pt, label={[label distance=-0.05cm, above]:$u_{i+1}$}] (C) at (1,1) {};
    \node[fill=black, circle, inner sep=1.5pt, label={[label distance=-0.05cm, above]:$u_{i}$}] (D) at (0,1) {};

    % Edges
    %\draw (C) -- (D) -- (A) -- (B);
    \draw (A) -- (B) -- (C) -- (D);
\end{tikzpicture}\caption{On the left is a 0-block, on the right is a 1-block. Note that $i$ is odd.}\label{fig:0-1-block}
\end{figure}

There are therefore four possibilities for the composition of two adjacent blocks: 00, 01, 10, and 11, as visualized in \Cref{fig:blocktypes}. This will arise in the analysis of $M_1$. 

\begin{figure}[htb!]
    \centering 

    \begin{tikzpicture}
    % Vertices
    \node[fill=black, circle, inner sep=1.5pt] (A) at (0,0) {};
    \node[fill=black, circle, inner sep=1.5pt] (B) at (1,0) {};
    \node[fill=black, circle, inner sep=1.5pt] (C) at (1,1) {};
    \node[fill=black, circle, inner sep=1.5pt] (D) at (0,1) {};

    % Edges
    \draw (C) -- (D) -- (A) -- (B);
\end{tikzpicture}\hspace*{10mm}
    \begin{tikzpicture}
    % Vertices
    \node[fill=black, circle, inner sep=1.5pt] (A) at (0,0) {};
    \node[fill=black, circle, inner sep=1.5pt] (B) at (1,0) {};
    \node[fill=black, circle, inner sep=1.5pt] (C) at (1,1) {};
    \node[fill=black, circle, inner sep=1.5pt] (D) at (0,1) {};

    % Edges
   \draw (C) -- (D) -- (A) -- (B);
\end{tikzpicture}\hspace{25mm}
    \begin{tikzpicture}
    % Vertices
    \node[fill=black, circle, inner sep=1.5pt] (A) at (0,0) {};
    \node[fill=black, circle, inner sep=1.5pt] (B) at (1,0) {};
    \node[fill=black, circle, inner sep=1.5pt] (C) at (1,1) {};
    \node[fill=black, circle, inner sep=1.5pt] (D) at (0,1) {};

    % Edges
    %\draw (A) -- (B) -- (C) -- (D);
    \draw (C) -- (D) -- (A) -- (B);
\end{tikzpicture}\hspace*{10mm}
    \begin{tikzpicture}
    % Vertices
    \node[fill=black, circle, inner sep=1.5pt] (A) at (0,0) {};
    \node[fill=black, circle, inner sep=1.5pt] (B) at (1,0) {};
    \node[fill=black, circle, inner sep=1.5pt] (C) at (1,1) {};
    \node[fill=black, circle, inner sep=1.5pt] (D) at (0,1) {};

    % Edges
    %\draw (C) -- (D) -- (A) -- (B);
    \draw (A) -- (B) -- (C) -- (D);
\end{tikzpicture}\vspace*{10mm}

    \begin{tikzpicture}
    % Vertices
    \node[fill=black, circle, inner sep=1.5pt] (A) at (0,0) {};
    \node[fill=black, circle, inner sep=1.5pt] (B) at (1,0) {};
    \node[fill=black, circle, inner sep=1.5pt] (C) at (1,1) {};
    \node[fill=black, circle, inner sep=1.5pt] (D) at (0,1) {};

    % Edges
    \draw (A) -- (B) -- (C) -- (D);
\end{tikzpicture}\hspace*{10mm}
    \begin{tikzpicture}
    % Vertices
    \node[fill=black, circle, inner sep=1.5pt] (A) at (0,0) {};
    \node[fill=black, circle, inner sep=1.5pt] (B) at (1,0) {};
    \node[fill=black, circle, inner sep=1.5pt] (C) at (1,1) {};
    \node[fill=black, circle, inner sep=1.5pt] (D) at (0,1) {};

    % Edges
    \draw (C) -- (D) -- (A) -- (B);
\end{tikzpicture}\hspace{25mm}
    \begin{tikzpicture}
    % Vertices
    \node[fill=black, circle, inner sep=1.5pt] (A) at (0,0) {};
    \node[fill=black, circle, inner sep=1.5pt] (B) at (1,0) {};
    \node[fill=black, circle, inner sep=1.5pt] (C) at (1,1) {};
    \node[fill=black, circle, inner sep=1.5pt] (D) at (0,1) {};

    % Edges
    \draw (A) -- (B) -- (C) -- (D);
\end{tikzpicture}\hspace*{10mm}
    \begin{tikzpicture}
    % Vertices
    \node[fill=black, circle, inner sep=1.5pt] (A) at (0,0) {};
    \node[fill=black, circle, inner sep=1.5pt] (B) at (1,0) {};
    \node[fill=black, circle, inner sep=1.5pt] (C) at (1,1) {};
    \node[fill=black, circle, inner sep=1.5pt] (D) at (0,1) {};

    % Edges
    \draw (A) -- (B) -- (C) -- (D);
\end{tikzpicture}\caption{The four configurations. Top left is 00, top right is 01, bottom left is 10, bottom right is 11.}\label{fig:blocktypes}
\end{figure}

\bedit{For every even $i \neq 0$}, we then take one edge from the pair $\{\{u_i,u_{i+1}\}, \{v_i,v_{i+1}\}\}$. Thus, \bedit{except for the pair of blocks containing $u_0$ and $u_1$ respectively}, each pair of adjacent blocks can be joined in two possible ways,
% thus joining these blocks in two possible ways,
creating eight possible configurations for adjacent blocks. This will be used in understanding the Eulerian subgraph resulting from adding $M_1$. $M_2$ will only involve a single block, and thus is simpler.

We will now examine the Eulerian graphs resulting from adding $M_1$ or $M_2$. $M_1$ will add edges between vertices of adjacent blocks, and $M_2$ will add edges between vertices in the same block. Thus, they create very different structures. 

\subsection{Bad Tours on $M_1$}

%A type 1 graph: a collection of cycles of length 2 (doubled edge), 5, or 8, ordered cyclically. 

We begin by considering the case where the matching being added is $M_1$. Recall $M_1 = \{(o_0,o_1),\allowbreak(o_2,o_3),\allowbreak\dots,\allowbreak(o_{2k-2},o_{2k-1})\}$, where $o_i$ is the vertex in $\{u_i, v_i\}$ with odd degree in the tree. As noted previously, this means that the matching edges are added \textit{between vertices of adjacent blocks}. This creates graphs of the type seen in \Cref{fig:type1}. 

\begin{figure}[htb!]
\centering
\begin{tikzpicture}[node distance=2cm, minimum size=4pt, every node/.style={circle, draw}, scale=0.9]
	\node[circle, inner sep=1pt] (node0) at (2.25, 0.0) {};
\node[circle, inner sep=1pt] (node16) at (3.0, 0.0) {};
\node[circle, inner sep=1pt] (node1) at (1.1250000000000004, 1.948557158514987) {};
\node[circle, inner sep=1pt] (node17) at (1.5000000000000004, 2.598076211353316) {};
\node[circle, inner sep=1pt] (node2) at (0.39070839975059346, 2.215817444277468) {};
\node[circle, inner sep=1pt] (node18) at (0.5209445330007912, 2.954423259036624) {};
\node[circle, inner sep=1pt] (node3) at (-0.3907083997505932, 2.215817444277468) {};
\node[circle, inner sep=1pt] (node19) at (-0.5209445330007909, 2.954423259036624) {};
\node[circle, inner sep=1pt] (node4) at (-1.1249999999999996, 1.948557158514987) {};
\node[circle, inner sep=1pt] (node20) at (-1.4999999999999996, 2.598076211353316) {};
\node[circle, inner sep=1pt] (node5) at (-1.7235999970177003, 1.4462721217947139) {};
\node[circle, inner sep=1pt] (node21) at (-2.2981333293569337, 1.9283628290596184) {};
\node[circle, inner sep=1pt] (node6) at (-2.1143083967682936, 0.769545322482755) {};
\node[circle, inner sep=1pt] (node22) at (-2.819077862357725, 1.0260604299770066) {};
\node[circle, inner sep=1pt] (node7) at (-2.25, 2.755455298081545e-16) {};
\node[circle, inner sep=1pt] (node23) at (-3.0, 3.6739403974420594e-16) {};
\node[circle, inner sep=1pt] (node8) at (-2.114308396768294, -0.7695453224827544) {};
\node[circle, inner sep=1pt] (node24) at (-2.8190778623577253, -1.026060429977006) {};
\node[circle, inner sep=1pt] (node9) at (-1.7235999970177005, -1.4462721217947134) {};
\node[circle, inner sep=1pt] (node25) at (-2.298133329356934, -1.9283628290596178) {};
\node[circle, inner sep=1pt] (node10) at (-1.1250000000000009, -1.9485571585149863) {};
\node[circle, inner sep=1pt] (node26) at (-1.5000000000000013, -2.598076211353315) {};
\node[circle, inner sep=1pt] (node11) at (-0.39070839975059324, -2.215817444277468) {};
\node[circle, inner sep=1pt] (node27) at (-0.520944533000791, -2.954423259036624) {};
\node[circle, inner sep=1pt] (node12) at (0.39070839975059246, -2.2158174442774685) {};
\node[circle, inner sep=1pt] (node28) at (0.5209445330007899, -2.9544232590366244) {};
\node[circle, inner sep=1pt] (node13) at (1.1250000000000004, -1.948557158514987) {};
\node[circle, inner sep=1pt] (node29) at (1.5000000000000004, -2.598076211353316) {};
\node[circle, inner sep=1pt] (node14) at (1.7235999970177, -1.4462721217947139) {};
\node[circle, inner sep=1pt] (node30) at (2.2981333293569333, -1.9283628290596186) {};
\node[circle, inner sep=1pt] (node15) at (2.114308396768293, -0.7695453224827563) {};
\node[circle, inner sep=1pt] (node31) at (2.8190778623577244, -1.0260604299770084) {};
\node[circle, inner sep=1pt, fill=black] (node32) at (2.4666931295630095, 0.8978028762298803) {};
\node[circle, inner sep=1pt, fill=black] (node33) at (2.0108666631873175, 1.6873174754271654) {};
\draw[color=black] (node1) -- (node2);
\draw[color=black] (node1) -- (node17);
\draw[color=red, dashed] (node2) -- (node3);
\draw[color=black] (node17) -- (node18);
\draw[color=black] (node17) -- (node33);
\draw[color=red, dashed] (node17) -- (node0);
\draw[color=black] (node18) -- (node19);
\draw[color=black] (node3) -- (node4);
\draw[color=black] (node4) -- (node20);
\draw[color=black] (node19) -- (node20);
\draw[color=black] (node20) -- (node21);
\draw[color=red,dashed] (node20) to[bend left=30] (node21);
\draw[color=black] (node5) -- (node6);
\draw[color=black] (node5) -- (node21);
\draw[color=red, dashed] (node6) -- (node23);
\draw[color=black] (node21) -- (node22);
\draw[color=black] (node22) -- (node23);
\draw[color=black] (node7) -- (node8);
\draw[color=black] (node7) -- (node23);
\draw[color=red, dashed] (node8) -- (node9);
\draw[color=black] (node23) -- (node24);
\draw[color=black] (node24) -- (node25);
\draw[color=black] (node9) -- (node10);
\draw[color=black] (node10) -- (node26);
\draw[color=black] (node25) -- (node26);
\draw[color=black] (node26) -- (node27);
\draw[color=red,dashed] (node26) to[bend left=30] (node27);
\draw[color=black] (node11) -- (node12);
\draw[color=black] (node11) -- (node27);
\draw[color=red, dashed] (node12) -- (node13);
\draw[color=black] (node27) -- (node28);
\draw[color=black] (node28) -- (node29);
\draw[color=black] (node13) -- (node14);
\draw[color=black] (node14) -- (node30);
\draw[color=black] (node29) -- (node30);
\draw[color=black] (node30) -- (node31);
\draw[color=red,dashed] (node30) to[bend left=30] (node31);
\draw[color=black] (node15) -- (node0);
\draw[color=black] (node15) -- (node31);
\draw[color=black] (node31) -- (node16);
\draw[color=black] (node16) -- (node32);
\draw[color=black] (node32) -- (node33);
\end{tikzpicture}\hspace*{10mm}
\begin{tikzpicture}[node distance=2cm, minimum size=14pt, every node/.style={circle, draw},scale=0.9]
\node[circle, inner sep=1pt] (node1) at (1.1250000000000004, 1.948557158514987) {$ 1 $};
\node[circle, inner sep=1pt] (node17) at (1.5000000000000004, 2.598076211353316) {$ 2 $};
\node[circle, inner sep=1pt] (node33) at (2.0108666631873175, 1.6873174754271654) {$ 3 $};
\node[circle, inner sep=1pt] (node32) at (2.4666931295630095, 0.8978028762298803) {$ 4 $};
\node[circle, inner sep=1pt] (node16) at (3.0, 0.0) {$ 5 $};
\node[circle, inner sep=1pt] (node31) at (2.8190778623577244, -1.0260604299770084) {$ 6 $};
\node[circle, inner sep=1pt] (node30) at (2.2981333293569333, -1.9283628290596186) {$ 7 $};
\node[circle, inner sep=1pt] (node14) at (1.7235999970177, -1.4462721217947139) {$ 8 $};
\node[circle, inner sep=1pt] (node13) at (1.1250000000000004, -1.948557158514987) {$ 9 $};
\node[circle, inner sep=1pt] (node12) at (0.39070839975059246, -2.2158174442774685) {$ 10 $};
\node[circle, inner sep=1pt] (node11) at (-0.39070839975059324, -2.215817444277468) {$ 11 $};
\node[circle, inner sep=1pt] (node27) at (-0.520944533000791, -2.954423259036624) {$ 12 $};
\node[circle, inner sep=1pt] (node26) at (-1.5000000000000013, -2.598076211353315) {$ 13 $};
\node[circle, inner sep=1pt] (node25) at (-2.298133329356934, -1.9283628290596178) {$ 14 $};
\node[circle, inner sep=1pt] (node24) at (-2.8190778623577253, -1.026060429977006) {$ 15 $};
\node[circle, inner sep=1pt] (node23) at (-3.0, 3.6739403974420594e-16) {$ 16 $};
\node[circle, inner sep=1pt] (node6) at (-2.1143083967682936, 0.769545322482755) {$ 17 $};
\node[circle, inner sep=1pt] (node5) at (-1.7235999970177003, 1.4462721217947139) {$ 18 $};
\node[circle, inner sep=1pt] (node21) at (-2.2981333293569337, 1.9283628290596184) {$ 19 $};
\node[circle, inner sep=1pt] (node20) at (-1.4999999999999996, 2.598076211353316) {$ 20 $};
\node[circle, inner sep=1pt] (node19) at (-0.5209445330007909, 2.954423259036624) {$ 21 $};
\node[circle, inner sep=1pt] (node18) at (0.5209445330007912, 2.954423259036624) {$ 22 $};
\node[circle, inner sep=1pt] (node0) at (2.25, 0.0) {$ 23 $};
\node[circle, inner sep=1pt] (node15) at (2.114308396768293, -0.7695453224827563) {$ 24 $};
\node[circle, inner sep=1pt] (node29) at (1.5000000000000004, -2.598076211353316) {$ 25 $};
\node[circle, inner sep=1pt] (node28) at (0.5209445330007899, -2.9544232590366244) {$ 26 $};
\node[circle, inner sep=1pt] (node10) at (-1.1250000000000009, -1.9485571585149863) {$ 27 $};
\node[circle, inner sep=1pt] (node9) at (-1.7235999970177005, -1.4462721217947134) {$ 28 $};
\node[circle, inner sep=1pt] (node8) at (-2.114308396768294, -0.7695453224827544) {$ 29 $};
\node[circle, inner sep=1pt] (node7) at (-2.25, 2.755455298081545e-16) {$ 30 $};
\node[circle, inner sep=1pt] (node22) at (-2.819077862357725, 1.0260604299770066) {$ 31 $};
\node[circle, inner sep=1pt] (node4) at (-1.1249999999999996, 1.948557158514987) {$ 32 $};
\node[circle, inner sep=1pt] (node3) at (-0.3907083997505932, 2.215817444277468) {$ 33 $};
\node[circle, inner sep=1pt] (node2) at (0.39070839975059346, 2.215817444277468) {$ 34 $};
\draw[color=black] (node1) -- (node2);
\draw[color=black] (node1) -- (node17);
\draw[color=red, dashed] (node2) -- (node3);
\draw[color=black] (node17) -- (node18);
\draw[color=black] (node17) -- (node33);
\draw[color=red, dashed] (node17) -- (node0);
\draw[color=black] (node18) -- (node19);
\draw[color=black] (node3) -- (node4);
\draw[color=black] (node4) -- (node20);
\draw[color=black] (node19) -- (node20);
\draw[color=black] (node20) -- (node21);
\draw[color=red,dashed] (node20) to[bend left=30] (node21);
\draw[color=black] (node5) -- (node6);
\draw[color=black] (node5) -- (node21);
\draw[color=red, dashed] (node6) -- (node23);
\draw[color=black] (node21) -- (node22);
\draw[color=black] (node22) -- (node23);
\draw[color=black] (node7) -- (node8);
\draw[color=black] (node7) -- (node23);
\draw[color=red, dashed] (node8) -- (node9);
\draw[color=black] (node23) -- (node24);
\draw[color=black] (node24) -- (node25);
\draw[color=black] (node9) -- (node10);
\draw[color=black] (node10) -- (node26);
\draw[color=black] (node25) -- (node26);
\draw[color=black] (node26) -- (node27);
\draw[color=red,dashed] (node26) to[bend left=30] (node27);
\draw[color=black] (node11) -- (node12);
\draw[color=black] (node11) -- (node27);
\draw[color=red, dashed] (node12) -- (node13);
\draw[color=black] (node27) -- (node28);
\draw[color=black] (node28) -- (node29);
\draw[color=black] (node13) -- (node14);
\draw[color=black] (node14) -- (node30);
\draw[color=black] (node29) -- (node30);
\draw[color=black] (node30) -- (node31);
\draw[color=red,dashed] (node30) to[bend left=30] (node31);
\draw[color=black] (node15) -- (node0);
\draw[color=black] (node15) -- (node31);
\draw[color=black] (node31) -- (node16);
\draw[color=black] (node16) -- (node32);
\draw[color=black] (node32) -- (node33);
\end{tikzpicture}
\caption{On the left is an example Eulerian graph when $M_1$ (dotted red edges) is added. The special nodes $w_0,w_1$ are marked in black. The graph consists of circuits of length 2, 5 or 8, except for the special circuit containing $w_0w_1$, which has length 4, 7, or 10. On the right is the Hamiltonian cycle resulting from shortcutting the adversarial Eulerian tour we construct here, in which we always alternate the side of the circuit we traverse. One can check that every shortcutting operation does not decrease the cost.}\label{fig:type1}
\end{figure}

\bedit{First, consider two adjacent blocks that are not adjacent to the envelope gadget. As discussed, there are eight possible configurations for the structure of the tree on these blocks.} For conciseness, we consider only the cases in which the edge on the outer ring is chosen to join the blocks, i.e. $(u_{2j},u_{2j+1})$ and not $(v_{2j},v_{2j+1})$. The other case is analyzed analogously up to a symmetry. This leaves four cases, exactly corresponding to the four possible orientations of the blocks $2j-1$ and $2j+1$ as discussed above. The four cases are illustrated in \Cref{fig:four_configs}.

\begin{figure}[htb!]
    \centering 
\begin{tikzpicture}
    % Vertices
    \node[fill=black, circle, inner sep=1.5pt] (A) at (0,0) {};
    \node[fill=red, circle, inner sep=1.5pt] (B) at (1,0) {};
    \node[fill=black, circle, inner sep=1.5pt,label={[label distance=-0.05cm, above]:$u_{2j}$}] (C) at (1,1) {};
    \node[fill=black, circle, inner sep=1.5pt] (D) at (0,1) {};

    \node[fill=black, circle, inner sep=1.5pt] (A') at (2,0) {};
    \node[fill=black, circle, inner sep=1.5pt] (B') at (3,0) {};
    \node[fill=black, circle, inner sep=1.5pt] (C') at (3,1) {};
    \node[fill=red, circle, inner sep=1.5pt,label={[label distance=-0.05cm, above]:$u_{2j+1}$}] (D') at (2,1) {};

    % Edges
    \draw (C) -- (D) -- (A) -- (B);
    \draw (C') -- (D') -- (A') -- (B');
    \draw[blue, thick] (C) -- (D');
    \draw[red,dashed] (B) -- (D');
\end{tikzpicture}\hspace{25mm}
\begin{tikzpicture}
    % Vertices
    \node[fill=black, circle, inner sep=1.5pt] (A) at (0,0) {};
    \node[fill=red, circle, inner sep=1.5pt] (B) at (1,0) {};
    \node[fill=black, circle, inner sep=1.5pt,label={[label distance=-0.05cm, above]:$u_{2j}$}] (C) at (1,1) {};
    \node[fill=black, circle, inner sep=1.5pt] (D) at (0,1) {};

    \node[fill=red, circle, inner sep=1.5pt] (A') at (2,0) {};
    \node[fill=black, circle, inner sep=1.5pt] (B') at (3,0) {};
    \node[fill=black, circle, inner sep=1.5pt] (C') at (3,1) {};
    \node[fill=black, circle, inner sep=1.5pt,label={[label distance=-0.05cm, above]:$u_{2j+1}$}] (D') at (2,1) {};

    % Edges
    \draw (C) -- (D) -- (A) -- (B);
    \draw (A') -- (B') -- (C') -- (D');
    \draw[blue, thick] (C) -- (D');
    \draw[red,dashed] (B) -- (A');
\end{tikzpicture}

\vspace*{10mm}
\begin{tikzpicture}
    % Vertices
    \node[fill=black, circle, inner sep=1.5pt] (A) at (0,0) {};
    \node[fill=black, circle, inner sep=1.5pt] (B) at (1,0) {};
    \node[fill=red, circle, inner sep=1.5pt,label={[label distance=-0.05cm, above]:$u_{2j}$}] (C) at (1,1) {};
    \node[fill=black, circle, inner sep=1.5pt] (D) at (0,1) {};

    \node[fill=black, circle, inner sep=1.5pt] (A') at (2,0) {};
    \node[fill=black, circle, inner sep=1.5pt] (B') at (3,0) {};
    \node[fill=black, circle, inner sep=1.5pt] (C') at (3,1) {};
    \node[fill=red, circle, inner sep=1.5pt,label={[label distance=-0.05cm, above]:$u_{2j+1}$}] (D') at (2,1) {};

    % Edges
    \draw (A) -- (B) -- (C) -- (D);
    \draw (C') -- (D') -- (A') -- (B');
    \draw[blue, thick] (C) -- (D');
    \draw[color=red,dashed] (C) to[bend right=30] (D');
\end{tikzpicture}\hspace{25mm}
\begin{tikzpicture}
    % Vertices
    \node[fill=black, circle, inner sep=1.5pt] (A) at (0,0) {};
    \node[fill=black, circle, inner sep=1.5pt] (B) at (1,0) {};
    \node[fill=red, circle, inner sep=1.5pt,label={[label distance=-0.05cm, above]:$u_{2j}$}] (C) at (1,1) {};
    \node[fill=black, circle, inner sep=1.5pt] (D) at (0,1) {};

    \node[fill=red, circle, inner sep=1.5pt] (A') at (2,0) {};
    \node[fill=black, circle, inner sep=1.5pt] (B') at (3,0) {};
    \node[fill=black, circle, inner sep=1.5pt] (C') at (3,1) {};
    \node[fill=black, circle, inner sep=1.5pt,label={[label distance=-0.05cm, above]:$u_{2j+1}$}] (D') at (2,1) {};

    % Edges
    \draw (A) -- (B) -- (C) -- (D);
    \draw (A') -- (B') -- (C') -- (D');
    \draw[blue, thick] (C) -- (D');
    \draw[red,dashed] (C) -- (A');
\end{tikzpicture}\caption{\bedit{The possible configurations in which two adjacent blocks can be connected given that $(u_{2j},u_{2j+1})$ is in the 1-tree. Here, $j \neq 0$.} The odd nodes, $o_{2j},o_{2j+1}$ are highlighted in red, and in dashed red is the edge $M_1$ will add.}
\label{fig:four_configs}
\end{figure}

\begin{figure}[htb!]
    \centering 
\begin{tikzpicture}
    % Vertices
    \node[fill=black, circle, inner sep=1.5pt] (A) at (0,0) {};
    \node[fill=red, circle, inner sep=1.5pt] (B) at (1,0) {};
    \node[fill=black, circle, inner sep=1.5pt,label={[label distance=-0.05cm, above]:$u_{0}$}] (C) at (1,1) {};
    \node[fill=black, circle, inner sep=1.5pt] (D) at (0,1) {};

    \node[fill=black, circle, inner sep=1.5pt] (A') at (3,0) {};
    \node[fill=black, circle, inner sep=1.5pt] (B') at (4,0) {};
    \node[fill=black, circle, inner sep=1.5pt] (C') at (4,1) {};
    \node[fill=red, circle, inner sep=1.5pt,label={[label distance=-0.05cm, above]:$u_{1}$}] (D') at (3,1) {};

    \node[fill=black, circle, inner sep=1.5pt] (w0) at (1.5, 0.5) {};
    \node[fill=black, circle, inner sep=1.5pt] (w1) at (2.5, 0.5) {};

    % Edges
    \draw (w0) -- (w1);
    \draw (C) -- (D) -- (A) -- (B);
    \draw (C') -- (D') -- (A') -- (B');
    \draw[blue, thick] (C) -- (w0);
    \draw[blue,thick] (w1) -- (D');
    \draw[color=red,dashed] (B) to[bend right=45] (D');
\end{tikzpicture}\hspace{25mm}
\begin{tikzpicture}
    % Vertices
    \node[fill=black, circle, inner sep=1.5pt] (A) at (0,0) {};
    \node[fill=red, circle, inner sep=1.5pt] (B) at (1,0) {};
    \node[fill=black, circle, inner sep=1.5pt,label={[label distance=-0.05cm, above]:$u_{0}$}] (C) at (1,1) {};
    \node[fill=black, circle, inner sep=1.5pt] (D) at (0,1) {};

    \node[fill=red, circle, inner sep=1.5pt] (A') at (3,0) {};
    \node[fill=black, circle, inner sep=1.5pt] (B') at (4,0) {};
    \node[fill=black, circle, inner sep=1.5pt] (C') at (4,1) {};
    \node[fill=black, circle, inner sep=1.5pt,label={[label distance=-0.05cm, above]:$u_{1}$}] (D') at (3,1) {};

    \node[fill=black, circle, inner sep=1.5pt] (w0) at (1.5, 0.5) {};
    \node[fill=black, circle, inner sep=1.5pt] (w1) at (2.5, 0.5) {};

    % Edges
    \draw (w0) -- (w1);
    \draw (C) -- (D) -- (A) -- (B);
    \draw (A') -- (B') -- (C') -- (D');
    \draw[blue, thick] (C) -- (w0);
    \draw[blue,thick] (w1) -- (D');    \draw[red,dashed] (B) -- (A');
\end{tikzpicture}

\vspace*{5mm}
\begin{tikzpicture}
    % Vertices
    \node[fill=black, circle, inner sep=1.5pt] (A) at (0,0) {};
    \node[fill=black, circle, inner sep=1.5pt] (B) at (1,0) {};
    \node[fill=red, circle, inner sep=1.5pt,label={[label distance=-0.05cm, above]:$u_{0}$}] (C) at (1,1) {};
    \node[fill=black, circle, inner sep=1.5pt] (D) at (0,1) {};

    \node[fill=black, circle, inner sep=1.5pt] (A') at (3,0) {};
    \node[fill=black, circle, inner sep=1.5pt] (B') at (4,0) {};
    \node[fill=black, circle, inner sep=1.5pt] (C') at (4,1) {};
    \node[fill=red, circle, inner sep=1.5pt,label={[label distance=-0.05cm, above]:$u_{1}$}] (D') at (3,1) {};

    \node[fill=black, circle, inner sep=1.5pt] (w0) at (1.5, 0.5) {};
    \node[fill=black, circle, inner sep=1.5pt] (w1) at (2.5, 0.5) {};

    % Edges
    \draw (w0) -- (w1);
    \draw (A) -- (B) -- (C) -- (D);
    \draw (C') -- (D') -- (A') -- (B');
    \draw[blue, thick] (C) -- (w0);
    \draw[blue,thick] (w1) -- (D');
    \draw[color=red,dashed] (C) -- (D');
\end{tikzpicture}\hspace{25mm}
\begin{tikzpicture}
    % Vertices
    \node[fill=black, circle, inner sep=1.5pt] (A) at (0,0) {};
    \node[fill=black, circle, inner sep=1.5pt] (B) at (1,0) {};
    \node[fill=red, circle, inner sep=1.5pt,label={[label distance=-0.05cm, above]:$u_{0}$}] (C) at (1,1) {};
    \node[fill=black, circle, inner sep=1.5pt] (D) at (0,1) {};

    \node[fill=red, circle, inner sep=1.5pt] (A') at (3,0) {};
    \node[fill=black, circle, inner sep=1.5pt] (B') at (4,0) {};
    \node[fill=black, circle, inner sep=1.5pt] (C') at (4,1) {};
    \node[fill=black, circle, inner sep=1.5pt,label={[label distance=-0.05cm, above]:$u_{1}$}] (D') at (3,1) {};

    \node[fill=black, circle, inner sep=1.5pt] (w0) at (1.5, 0.5) {};
    \node[fill=black, circle, inner sep=1.5pt] (w1) at (2.5, 0.5) {};

    % Edges
    \draw (w0) -- (w1);
    \draw (D) -- (C) -- (B) -- (A);
    \draw (D') -- (C') -- (B') -- (A');
    \draw[blue, thick] (C) -- (w0);
    \draw[blue,thick] (w1) -- (D');    \draw[red,dashed] (C) to[bend right=45] (A');
\end{tikzpicture}

%%%%%
\vspace*{10mm}
%%%%%

\begin{tikzpicture}
    % Vertices
    \node[fill=black, circle, inner sep=1.5pt] (A) at (0,0) {};
    \node[fill=red, circle, inner sep=1.5pt] (B) at (1,0) {};
    \node[fill=black, circle, inner sep=1.5pt,label={[label distance=-0.05cm, above]:$u_{0}$}] (C) at (1,1) {};
    \node[fill=black, circle, inner sep=1.5pt] (D) at (0,1) {};

    \node[fill=red, circle, inner sep=1.5pt] (A') at (3,0) {};
    \node[fill=black, circle, inner sep=1.5pt] (B') at (4,0) {};
    \node[fill=black, circle, inner sep=1.5pt] (C') at (4,1) {};
    \node[fill=black, circle, inner sep=1.5pt,label={[label distance=-0.05cm, above]:$u_{1}$}] (D') at (3,1) {};

    \node[fill=black, circle, inner sep=1.5pt] (w0) at (1.5, 0.5) {};
    \node[fill=black, circle, inner sep=1.5pt] (w1) at (2.5, 0.5) {};

    % Edges
    \draw (w0) -- (w1);
    \draw (C) -- (D) -- (A) -- (B);
    \draw (C') -- (D') -- (A') -- (B');
    \draw[blue, thick] (C) -- (w0);
    \draw[blue,thick] (w1) -- (A');
    \draw[color=red,dashed] (B) -- (A');
\end{tikzpicture}\hspace{25mm}
\begin{tikzpicture}
    % Vertices
    \node[fill=black, circle, inner sep=1.5pt] (A) at (0,0) {};
    \node[fill=red, circle, inner sep=1.5pt] (B) at (1,0) {};
    \node[fill=black, circle, inner sep=1.5pt,label={[label distance=-0.05cm, above]:$u_{0}$}] (C) at (1,1) {};
    \node[fill=black, circle, inner sep=1.5pt] (D) at (0,1) {};

    \node[fill=black, circle, inner sep=1.5pt] (A') at (3,0) {};
    \node[fill=black, circle, inner sep=1.5pt] (B') at (4,0) {};
    \node[fill=black, circle, inner sep=1.5pt] (C') at (4,1) {};
    \node[fill=red, circle, inner sep=1.5pt,label={[label distance=-0.05cm, above]:$u_{1}$}] (D') at (3,1) {};

    \node[fill=black, circle, inner sep=1.5pt] (w0) at (1.5, 0.5) {};
    \node[fill=black, circle, inner sep=1.5pt] (w1) at (2.5, 0.5) {};

    % Edges
    \draw (w0) -- (w1);
    \draw (C) -- (D) -- (A) -- (B);
    \draw (A') -- (B') -- (C') -- (D');
    \draw[blue, thick] (C) -- (w0);
    \draw[blue,thick] (w1) -- (A');    \draw[red,dashed] (B) -- (D');
\end{tikzpicture}

\vspace*{5mm}
\begin{tikzpicture}
    % Vertices
    \node[fill=black, circle, inner sep=1.5pt] (A) at (0,0) {};
    \node[fill=black, circle, inner sep=1.5pt] (B) at (1,0) {};
    \node[fill=red, circle, inner sep=1.5pt,label={[label distance=-0.05cm, above]:$u_{0}$}] (C) at (1,1) {};
    \node[fill=black, circle, inner sep=1.5pt] (D) at (0,1) {};

    \node[fill=red, circle, inner sep=1.5pt] (A') at (3,0) {};
    \node[fill=black, circle, inner sep=1.5pt] (B') at (4,0) {};
    \node[fill=black, circle, inner sep=1.5pt] (C') at (4,1) {};
    \node[fill=black, circle, inner sep=1.5pt,label={[label distance=-0.05cm, above]:$u_{1}$}] (D') at (3,1) {};

    \node[fill=black, circle, inner sep=1.5pt] (w0) at (1.5, 0.5) {};
    \node[fill=black, circle, inner sep=1.5pt] (w1) at (2.5, 0.5) {};

    % Edges
    \draw (w0) -- (w1);
    \draw (A) -- (B) -- (C) -- (D);
    \draw (C') -- (D') -- (A') -- (B');
    \draw[blue, thick] (C) -- (w0);
    \draw[blue,thick] (w1) -- (A');
    \draw[color=red,dashed] (C) to[bend left=45] (A');
\end{tikzpicture}\hspace{25mm}
\begin{tikzpicture}
    % Vertices
    \node[fill=black, circle, inner sep=1.5pt] (A) at (0,0) {};
    \node[fill=black, circle, inner sep=1.5pt] (B) at (1,0) {};
    \node[fill=red, circle, inner sep=1.5pt,label={[label distance=-0.05cm, above]:$u_{0}$}] (C) at (1,1) {};
    \node[fill=black, circle, inner sep=1.5pt] (D) at (0,1) {};

    \node[fill=black, circle, inner sep=1.5pt] (A') at (3,0) {};
    \node[fill=black, circle, inner sep=1.5pt] (B') at (4,0) {};
    \node[fill=black, circle, inner sep=1.5pt] (C') at (4,1) {};
    \node[fill=red, circle, inner sep=1.5pt,label={[label distance=-0.05cm, above]:$u_{1}$}] (D') at (3,1) {};

    \node[fill=black, circle, inner sep=1.5pt] (w0) at (1.5, 0.5) {};
    \node[fill=black, circle, inner sep=1.5pt] (w1) at (2.5, 0.5) {};

    % Edges
    \draw (w0) -- (w1);
    \draw (D) -- (C) -- (B) -- (A);
    \draw (D') -- (C') -- (B') -- (A');
    \draw[blue, thick] (C) -- (w0);
    \draw[blue,thick] (w1) -- (A');    \draw[red,dashed] (C) -- (D');
\end{tikzpicture}
\caption{\bedit{The possible configurations for the two blocks adjacent to the envelope gadget, assuming $u_0w_0$ is in the tree. The top four figures show the possibilities when $w_1u_1$ is in the tree; the bottom four figures show the possibilities when $w_1v_1$ is in the tree. The odd nodes, $o_{0},o_{1}$ are highlighted in red, and in dashed red is the edge $M_1$ will add.}}
\label{fig:four_configs_gadget}
\end{figure}

\begin{enumerate}
    \item In Case 00, $o_{2j}=v_{2j}$ and $o_{2j+1}=u_{2j+1}$. Thus by adding the edge $(o_{2j},o_{2j+1})$ the block $2j-1$ becomes a circuit containing 5 vertices together with $u_{2j+1}$.
    \item In Case 01, $o_{2j}=v_{2j}$ and $o_{2j+1}=v_{2j+1}$. Therefore adding the edge $(o_{2j},o_{2j+1})$ creates a circuit of length 8 containing the two blocks $2j-1,2j+1$. 
    \item In Case 10, $o_{2j}=u_{2j}$ and $o_{2j+1}=u_{2j+1}$. Therefore adding the edge $(o_{2j},o_{2j+1})$ joins the blocks $2j-1$ and $2j+1$ with a doubled edge, i.e. creates the circuit $\{o_{2j},o_{2j+1}\}$
    \item In Case 11, $o_{2j}=u_{2j}$ and $o_{2j+1}=v_{2j+1}$. Therefore adding the edge $(o_{2j},o_{2j+1})$ creates a circuit of length 5 containing the block $2j+1$ and $u_{2j}$.
\end{enumerate}

\bedit{Analogously, we analyze the possible configurations for the two blocks that are adjacent the envelope gadget. All configurations are shown in \Cref{fig:four_configs_gadget}, assuming that $u_0w_0$ is included in the tree. (The configurations where $v_0w_0$ is included are symmetric.) From the diagrams, we observe that $w_0w_1$ is in a circuit of length 4, 7, or 10.}

\bedit{Therefore the resulting graph $T \cup M_1$ consists of a collection of circuits arranged in a circle. These circuits have length 2, 5, and 8, except for the circuit containing $w_0w_1$, which has length 4, 7, or 10. We will refer to the circuit containing $w_0w_1$ as the "special circuit" from now on.} Note that doubled edges (circuits of length 2) are never adjacent to one another on this circle, since they are only created between vertices $(o_i,o_{i+1})$ for even values of $i$. 

We now describe the problematic tours on such graphs—namely, Eulerian tours for which the resulting Hamiltonian cycle obtained via shortcutting is no cheaper. For intuition, the reader may first consider the bad tour shown in \Cref{fig:type1}. 
% We begin at an arbitrary vertex $t_0$ of degree 2 and pick an edge in the clockwise direction. 
\bedit{
We begin by considering the special circuit. This circuit has two vertices of degree 4 -- one located counterclockwise from $w_0$, and the other clockwise. Let $t_0$ be the vertex in this circuit that lies counterclockwise from $w_0$, and designate it as the starting point of the tour. Note that every vertex in the tour has degree 2 or 4. Each degree 2 vertex is visited once. Each degree 4 vertex is visited twice, except for $t_0$ which is visited three times because the tour starts and ends at $t_0$.}

\bedit{We now describe the procedure for picking the next vertex $t_{k+1}$ given our current vertex $t_k$. 
If $t_k$ has degree 2, there is no choice to make, as only one edge remains. So it is sufficient to describe decisions on vertices $t_k$ of degree 4.
% If $t_k$ has degree 2 (or has degree 4 but is being visited for the second time), there is no choice to make, as only one edge remains. So it is sufficient to describe decisions on vertices $t_k$ of degree 4 visited for the first time. 
Note that since $t_k$ has degree 4, it is at the intersection of two adjacent circuits. Let  $C$ denote the circuit in the clockwise direction adjacent to $t_k$, and let $C'$ denote the circuit in the counterclockwise direction adjacent to $t_k$. 
% Note that by our choice of $t_0$, $C$ is not the special circuit. Thus, $C$ is either a doubled edge or has length 5 or 8. 
The next edge to traverse is determined by the following rules, in order of priority:}
% In the following, we slightly abuse notation and will refer to the fact that either $w=u$ or $w=v$ so that $w_i=u_i$ or $w_i=v_i$. 
\begin{enumerate}
    \item \bedit{\textbf{If $t_k = t_0$, pick an arbitrary edge in $C$ that has not been traversed yet.} Note that if $t_k \neq t_0$ and $t_k$ is being visited for the second time, there is no choice to make as only one edge remains. Thus for the remaining rules we assume that $t_k \neq t_0$  and $t_k$ is a vertex of degree 4 being visited for the first time.}
    \item \textbf{Never traverse an edge in the counterclockwise direction.} Therefore, if $C$ is a doubled edge, we immediately traverse one of its edges.
    \item \textbf{Alternate the visited side of adjacent circuits.} \bedit{Otherwise, $t_k \neq t_0$ and $C$ has length 5 or 8.} For simplicity, suppose $t_k$ is on the outer ring; the case where $t_k$ is on the inner ring is symmetric. Let $e_{\mathrm{outer}} = \{t_k, u\}$ and $e_{\mathrm{inner}} = \{t_k, v\}$ be the two edges in $C$ adjacent to $t_k$, where $u$ is on the outer ring and $v$ is on the inner ring. Let $e = \{t_{j}, t_{j+1}\}$ be the previous edge in the tour that was 
    %not part of a circuit of length 2. 
    \bedit{part of a circuit of length 5 or 8.}
    % Thus, $j=k-1$, $e = \{t_{k-1}, t_k\}$  if $C'$ is of length 5 or 8, and $j=k-2$, $e = \{t_{k-2}, t_{k-1}\}$ if $C'$ is of length 2. 
    Now, if $t_j$ is on the outer ring, we take $e_{\mathrm{inner}}$. Otherwise, take $e_{\mathrm{outer}}$. The intuition here is that if we visited the inner ring while traversing the last circuit of length greater than 2 \bedit{that was not the special circuit}, we now wish to visit the outer ring, and vice versa.
\end{enumerate}
% \begin{enumerate}
%     \item \textbf{Never traverse an edge in the counterclockwise direction.} Therefore, if there is a doubled edge adjacent to $w_i$ equal to $(w_i,u_{i+1})$ or $(w_i,v_{i+1})$, we immediately traverse it.
%     \item \textbf{Alternate the unvisited side of adjacent circuits.} 
    
%     Otherwise, there are two available edges that do not go in the counterclockwise direction, $e=(w_i,w'_{i'})$ and $f=(w_i,w_{i+1})$ (where $w,w' \in \{u,v\}$ and $i' \in \{i,i+1\}$). $e,f$ are then part of a circuit of length 5 or length 8. Inspect the previous circuit $C$ visited by the tour that was of length not equal to 2. Since this is the first time $w_i$ was visited, and $w_i$ is of degree 4, it must be that some vertex $w''_j$ in $C$ has not yet been visited. If $w'=w''$ then take the edge $e=(w_i,w'_{i'})$. Otherwise, take the edge $f$.  The intuition here is that if in the last circuit we visited the inner ring, we now wish to visit the outer ring, and vice versa.
% \end{enumerate}

We call the resulting Eulerian tours $B$-tours because they are bad for the objective function. 

Let $R=t_0,\dots,t_m$ be a $B$-tour and let \bedit{$h_0,\dots,h_{4k+1}, h_0$} be the order the vertices are visited in the Hamiltonian cycle resulting from shortcutting $R$. 
% Note that unlike the above, we do not abuse notation and it is not necessarily the case that $t_i,h_i \in \{u_i,v_i\}$. 

\begin{lemma}
    Shortcutting a $B$-tour on a graph $T \cup M_1$ to a Hamiltonian cycle does not reduce the cost by more than 9. 
\end{lemma}
\begin{proof}
A shortcut occurs when we are about to arrive at a vertex $t_i$ of degree 4 for the second time. Let $t_{i-1} = h_j$ for some $j$ and $h_{j+1} = t_{i+\ell}$ for some $\ell$. 

We first observe that $\ell \le 2$. This is because in the graph $T \cup M_1$, as observed above, there are no adjacent circuits of length 2. Therefore, there are no paths of vertices of degree 4 containing more than 2 vertices. 

% Furthermore, if $\ell=2$ then $t_i,t_{i+1}$ are connected by a circuit of length 2, i.e. a doubled edge. 

We assume first that \bedit{$t_i$ and $t_{i+1}$ are not in the special circuit and} $t_{i+\ell} \not= t_0$. In this scenario, we show the shortcutting does not reduce the cost at all.  We argue about two cases:
\begin{enumerate}
    \item $\ell=1$, i.e. we only skip vertex $t_i$.
    
    First suppose that $c(t_{i-1},t_{i})+c(t_{i},t_{i+1})=2$. Then, skipping $t_i$ reduces the cost if and only if $c(t_{i-1}, t_{i+1}) = 1$.
    \bedit{This implies that all pairwise distances among $t_{i-1}$, $t_i$, and $t_{i+1}$ are equal to 1.  In the $k$-donut graph, the only such triples are the two triangles $\{u_0, v_0, w_0\}$ and $\{u_1, v_1, w_1\}$. However, this is impossible because the special circuit contains at least two vertices from each triple, and we assumed that $t_i$ and $t_{i+1}$ are not in the special circuit.}
    % However, since the underlying $k$-donut graph is bipartite (and in particular, triangle free), there cannot exist length 1 edges $(t_{i-1},t_i),(t_i,t_{i+1}),(t_{i-1},t_{i+1})$, which finishes this case. 
    
    Therefore, it remains to deal with the case where $c(t_{i-1},t_{i})+c(t_{i},t_{i+1})=3$. First note that if $t_{i-1},t_{i+1}$ lie on different rings then $c(t_{i-1},t_{i+1})=3$ as desired. However, this must be the case, due to the fact that we always alternate the visited side of adjacent circuits. Indeed, since $\ell=1$, the circuit $C$ containing $t_{i}$ and $t_{i+1}$ and the circuit $C'$ containing $t_{i-1}$ and $t_i$  both have length greater than 2. Therefore, letting $j$ denote the index of the tour when $t_i$ was visited for the first time (so $j < i$ and $t_j = t_i$), we know from the rule that $t_{j-1}$ and $t_{j+1}$ are on different rings. This implies that $t_{i-1}$ and $t_{i+1}$ are also on different rings.
    % The circuit $C$ containing $t_{i-1},t_i$ must be of length greater than 2. Therefore when $t_i$ was visited for the first time, $C$ was the last visited circuit of length larger than 2, so our rule would identify a vertex on the ring containing $t_{i-1}$ as the unvisited vertex, implying that $t_{i+1}$ and $t_{i-1}$ are not on the same ring. 
    See \Cref{fig:skip1} for an illustration of this case.

    \begin{figure}[ht]
        \centering

\begin{tikzpicture}
    % Vertices
    \node[fill=black, circle, inner sep=1.5pt] (A) at (0,0) {};
    \node[fill=black, circle, inner sep=1.5pt,label={[blue,thick,below]:$t_{i-1}$}] (B) at (1,0) {};
    \node[fill=black, circle, inner sep=1.5pt,,label={[blue,thick,above]:$t_{i}$}] (C) at (1,1) {};
    \node[fill=black, circle, inner sep=1.5pt] (D) at (0,1) {};

    \node[fill=black, circle, inner sep=1.5pt,] (A') at (2,0) {};
    \node[fill=black, circle, inner sep=1.5pt] (B') at (3,0) {};
    \node[fill=black, circle, inner sep=1.5pt] (C') at (3,1) {};
    \node[fill=black, circle, inner sep=1.5pt,,label={[blue,thick]:$t_{i+1}$}] (D') at (2,1) {};

    % Edges
    \draw (A) -- (B) -- (C) -- (D);
	\draw (C) -- (D') -- (C');
	\draw (C) -- (A') -- (B');
    
    \draw[red,thick,<-,dotted] (C) to[bend right=20] (D);
    \draw[red,thick,->,dotted] (C) to[bend left=20] (A');
    \draw[red,thick,->,dotted] (A') to[bend left=20] (B');
    \draw[blue,thick,<-,dashed] (B) to[bend right=20] (A);
    \draw[blue,thick,->,dashed] (B) to (D');
    \draw[blue,thick,->,dashed] (D') to[bend right=20] (C');
\end{tikzpicture}\hspace{25mm}
\begin{tikzpicture}
    % Vertices
    \node[fill=black, circle, inner sep=1.5pt] (A) at (0,0) {};
    \node[fill=black, circle, inner sep=1.5pt] (B) at (1,0) {};
    \node[fill=black, circle, inner sep=1.5pt,label={[blue,thick]:$t_{i-1}$}] (C) at (1,1) {};
    \node[fill=black, circle, inner sep=1.5pt] (D) at (0,1) {};

    \node[fill=black, circle, inner sep=1.5pt,label={[blue,thick,below]:$t_{i+1}$}] (A') at (2,0) {};
    \node[fill=black, circle, inner sep=1.5pt,] (B') at (3,0) {};
    \node[fill=black, circle, inner sep=1.5pt,] (C') at (3,1) {};
    \node[fill=black, circle, inner sep=1.5pt,label={[blue,thick]:$t_{i}$}] (D') at (2,1) {};

    % Edges
    \draw (C) -- (D);
    \draw (A) -- (B);
    \draw (C') -- (D') -- (A') -- (B');
    \draw (C) -- (D');
    \draw (B) -- (D');
    
    \draw[red,thick,->,dotted] (D') to[bend left=20] (C');
    \draw[red,thick,->,dotted] (A) to[bend left=20] (B);
    \draw[red,thick,->,dotted] (B) to[bend left=20] (D');
    
    \draw[blue,thick,->,dashed] (D) to[bend left=20] (C);
    \draw[blue,thick,->,dashed] (C) to (A');
    \draw[blue,thick,->,dashed] (A') to[bend left=20] (B');
\end{tikzpicture}

\vspace*{5mm}
\begin{tikzpicture}
    % Vertices
    \node[fill=black, circle, inner sep=1.5pt] (A) at (0,0) {};
    \node[fill=black, circle, inner sep=1.5pt,label={[blue,thick,below]:$t_{i-1}$}] (B) at (1,0) {};
    \node[fill=black, circle, inner sep=1.5pt] (C) at (1,1) {};
    \node[fill=black, circle, inner sep=1.5pt] (D) at (0,1) {};

    \node[fill=black, circle, inner sep=1.5pt,] (A') at (2,0) {};
    \node[fill=black, circle, inner sep=1.5pt] (B') at (3,0) {};
    \node[fill=black, circle, inner sep=1.5pt,label={[blue,thick]:$t_{i+1}$}] (C') at (3,1) {};
    \node[fill=black, circle, inner sep=1.5pt,,label={[blue,thick]:$t_{i}$}] (D') at (2,1) {};

    % Edges
    \draw (C) -- (D);
    \draw[red,thick,<-,dotted] (C) to[bend right=20] (D);
    \draw[red,thick,<-,dotted] (D') to[bend right=20] (C);
    \draw[red,thick,->,dotted] (D') to[bend right=20] (A');
    \draw[red,thick,->,dotted] (A') to[bend right=20] (B');
    \draw[blue,thick,->,dashed] (A) to[bend left=20] (B);
    \draw[blue,thick,->,dashed] (B) to (C');
    \draw (A) -- (B);
    \draw (C') -- (D') -- (A') -- (B');
    \draw (C) -- (D');
    \draw (B) -- (D');
\end{tikzpicture}\hspace{25mm}
\begin{tikzpicture}
    % Vertices
    \node[fill=black, circle, inner sep=1.5pt] (A) at (0,0) {};
    \node[fill=black, circle, inner sep=1.5pt] (B) at (1,0) {};
    \node[fill=black, circle, inner sep=1.5pt,label={[blue,thick]:$t_{i}$}] (C) at (1,1) {};
    \node[fill=black, circle, inner sep=1.5pt,,label={[blue,thick]:$t_{i-1}$}] (D) at (0,1) {};

    \node[fill=black, circle, inner sep=1.5pt,label={[blue,thick,below]:$t_{i+1}$}] (A') at (2,0) {};
    \node[fill=black, circle, inner sep=1.5pt,] (B') at (3,0) {};
    \node[fill=black, circle, inner sep=1.5pt,] (C') at (3,1) {};
    \node[fill=black, circle, inner sep=1.5pt] (D') at (2,1) {};

    % Edges
    \draw (A) -- (B) -- (C) -- (D);
	\draw (C) -- (D') -- (C');
	\draw (C) -- (A') -- (B');
    
    \draw[red,thick,->,dotted] (D') to[bend left=20] (C');
    \draw[red,thick,->,dotted] (C) to[bend left=20] (D');
    \draw[red,thick,->,dotted] (A) to[bend left=20] (B);
    \draw[red,thick,->,dotted] (B) to[bend left=20] (C);
    
    \draw[blue,thick,->,dashed] (D) to (A');
    \draw[blue,thick,->,dashed] (A') to[bend left=20] (B');
\end{tikzpicture}

        \caption{Case 1: Skipping one vertex. In dotted red is the first visit, in dashed blue the second. The two figures on the top illustrate the case where $c(t_{i-1}, t_i) + c(t_i, t_{i+1}) = 2$. The two figures on the bottom illustrate the case where $c(t_{i-1}, t_i) + c(t_i, t_{i+1}) = 3$.}
        \label{fig:skip1}
    \end{figure}

    \item $\ell=2$, i.e. we skip vertices $t_i,t_{i+1}$. In this case, \bedit{since $t_i$ and $t_{i+1}$ are not in the special circuit}, they must be connected by a doubled edge and are therefore on the same ring. Therefore, there are no edges of length 2 in the path $(t_{i-1},t_i,t_{i+1},t_{i+2})$. So, we need to prove that $c(t_{i-1},t_{i+2}) = c(t_{i-1},t_{i})+c(t_{i},t_{i+1})+c(t_{i+1},t_{i+2})=3$.  First note that if one of $t_{i-1},t_{i+2}$ lies on the inner ring and one lies on the outer ring then $c(t_{i-1},t_{i+2})=3$. However similar to above this must be the case. See \Cref{fig:skip2} for an illustration.
    
    \begin{figure}[ht]
        \centering
        \begin{tikzpicture}
    % Vertices
    \node[fill=black, circle, inner sep=1.5pt] (A) at (0,0) {};
    \node[fill=black, circle, inner sep=1.5pt,label={[thick,blue,below]:$t_{i-1}$}] (B) at (1,0) {};
    \node[fill=black, circle, inner sep=1.5pt,,label={[thick,blue,above]:$t_{i}$}] (C) at (1,1) {};
    \node[fill=black, circle, inner sep=1.5pt] (D) at (0,1) {};

    \node[fill=black, circle, inner sep=1.5pt,] (A') at (2,0) {};
    \node[fill=black, circle, inner sep=1.5pt] (B') at (3,0) {};
    \node[fill=black, circle, inner sep=1.5pt,label={[thick,blue]:$t_{i+2}$}] (C') at (3,1) {};
    \node[fill=black, circle, inner sep=1.5pt,label={[thick,blue]:$t_{i+1}$}] (D') at (2,1) {};

    % Edges
    \draw (A) -- (B) -- (C) -- (D);
	\draw (B') -- (A') -- (D') -- (C');
	\draw (C) -- (D');
    
    \draw[thick,red,<-,dotted] (C) to[bend right=20] (D);
    \draw[thick,red,->,dotted] (C) to[bend left=20] (D');
    \draw[thick,red,->,dotted] (D') to[bend left=20] (A');
    \draw[thick,red,->,dotted] (A') to[bend left=20] (B');

    \draw[thick,blue,<-,dashed] (B) to[bend right=20] (A);
    \draw[thick,blue,->,dashed] (B) to (C');
\end{tikzpicture}\hspace{25mm}
\begin{tikzpicture}
    % Vertices
    \node[fill=black, circle, inner sep=1.5pt] (A) at (0,0) {};
    \node[fill=black, circle, inner sep=1.5pt] (B) at (1,0) {};
    \node[fill=black, circle, inner sep=1.5pt,label={[thick,blue]:$t_{i}$}] (C) at (1,1) {};
    \node[fill=black, circle, inner sep=1.5pt,,label={[thick,blue]:$t_{i-1}$}] (D) at (0,1) {};

    \node[fill=black, circle, inner sep=1.5pt,label={[thick,blue,below]:$t_{i+2}$}] (A') at (2,0) {};
    \node[fill=black, circle, inner sep=1.5pt,] (B') at (3,0) {};
    \node[fill=black, circle, inner sep=1.5pt] (C') at (3,1) {};
    \node[fill=black, circle, inner sep=1.5pt,label={[thick,blue,above]:$t_{i+1}$}] (D') at (2,1) {};

    % Edges
    \draw (A) -- (B) -- (C) -- (D);
	\draw (B') -- (A') -- (D') -- (C');
	\draw (C) -- (D');
    
    \draw[thick,red,->,dotted] (D') to[bend left=20] (C');
    \draw[thick,red,->,dotted] (C) to[bend left=20] (D');
    \draw[thick,red,->,dotted] (A) to[bend left=20] (B);
    \draw[thick,red,->,dotted] (B) to[bend left=20] (C);
    
    \draw[thick,blue,->,dashed] (D) to (A');
    \draw[thick,blue,->,dashed] (A') to[bend left=20] (B');
\end{tikzpicture}

        \caption{Case 2: Skipping two vertices. In dotted red is the first visit, in dashed blue the second.}
        \label{fig:skip2}
    \end{figure}
\end{enumerate}

If \bedit{$t_{i+\ell} = t_0$, or either $t_i$ or $t_{i+1}$ is on the special circuit,} we allow for the possibility that the shortcutting succeeded in reducing cost. \bedit{Note that this type of shortcutting can happen at most 3 times. Since $\ell \le 2$ and there are no adjacent edges of length 2, the cost of the tour is reduced by at most 3 each time such a shortcutting occurs. Thus, the total reduction in the cost is at most 9.}
\end{proof}

The above lemma therefore demonstrates that the asymptotic cost of the tour after shortcutting in this manner is still $11/8$ times that of the optimal tour.

%At least one of the cycles adjacent to $w$ is not of length 2. We will move in the direction of this cycle (if they are both not of length 2, pick one arbitrarily). If the cycle has length 

\subsection{Bad Tours on $M_2$}

We now consider the case where the matching being added is $M_2$. Recall $M_2 = \{(o_{2k-1},o_0),\allowbreak(o_1,o_2),\allowbreak\dots,\allowbreak(o_{2k-3},o_{2k-2})\}$. This indicates that the matching edges are added \textit{between vertices of the same block.} \bedit{In particular, no matching edges are added between two vertices in the envelope gadget.}
% \textit{between vertices on a common square of solid and dotted edges}.
We assume that the block is of type 0 (as the other case is the same up to a symmetry) consisting of vertices $u_i,u_{i+1},v_i,v_{i+1}$. Since the block is type 0, the edges $(u_i,u_{i+1}),(v_i,v_{i+1}),(u_i,v_i)$ are those present in the block. Thus what matters is the edges chosen among the pairs $\{(u_{i-1},u_i),(v_{i-1},v_i)\}$ and $\{(u_{i+1},u_{i+2}),(v_{i+1},v_{i+2})\}$ (see \Cref{fig:m2config}):
\begin{enumerate}
    \item If $(u_{i-1},u_i)$ and 
$(u_{i+1},u_{i+2})$
    % $(u_{i},u_{i+1})$ 
    are chosen, then $u_i$ and $v_{i+1}$ are odd, and we create a triangle by adding $(u_i,v_{i+1})$ in the matching. Similarly, if $(v_{i-1},v_i)$ and 
    $(v_{i+1},v_{i+2})$
    % $(v_{i},v_{i+1})$ 
    are chosen, we get a triangle via adding the edge $(v_i,u_{i+1})$. 
    \item If $(u_{i-1},u_i)$ and 
    $(v_{i+1},v_{i+2})$
    % $(v_{i},v_{i+1})$ 
    are chosen, then $u_i$ and $u_{i+1}$ are odd, and thus we add a second edge $(u_i,u_{i+1})$. Similarly if $(v_{i-1},v_i)$ and 
    $(u_{i+1},u_{i+2})$
    % $(u_{i},u_{i+1})$ 
    are added, then we add a second edge $(v_i,v_{i+1})$. 
\end{enumerate}
\begin{figure}[htb!]
    \centering 
\begin{tikzpicture}
    % Vertices
    \node[fill=red, circle, inner sep=1.5pt, label={[label distance=-0.15cm, below]:$v_{i}$}] (A) at (0,0) {};
    \node[fill=black, circle, inner sep=1.5pt, label={[label distance=-0.15cm, below]:$v_{i+1}$}] (B) at (1,0) {};
    \node[fill=red, circle, inner sep=1.5pt, label={[label distance=-0.05cm, above]:$u_{i+1}$}] (C) at (1,1) {};
    \node[fill=black, circle, inner sep=1.5pt, label={[label distance=-0.05cm, above]:$u_{i}$}] (D) at (0,1) {};

    % Edges
    \draw (C) -- (D) -- (A) -- (B);
    \draw[blue] (-1,0) -- (A);
    \draw[blue] (B) -- (2,0);
    \draw[red,dashed] (A) -- (C);
\end{tikzpicture}\hspace{25mm}
\begin{tikzpicture}
    % Vertices
    \node[fill=red, circle, inner sep=1.5pt, label={[label distance=-0.15cm, below]:$v_{i}$}] (A) at (0,0) {};
    \node[fill=red, circle, inner sep=1.5pt, label={[label distance=-0.15cm, below]:$v_{i+1}$}] (B) at (1,0) {};
    \node[fill=black, circle, inner sep=1.5pt, label={[label distance=-0.05cm, above]:$u_{i+1}$}] (C) at (1,1) {};
    \node[fill=black, circle, inner sep=1.5pt, label={[label distance=-0.05cm, above]:$u_{i}$}] (D) at (0,1) {};

    % Edges
    \draw (C) -- (D) -- (A) -- (B);
    \draw[blue] (-1,0) -- (A);
    \draw[blue] (C) -- (2,1);
    \draw[red,dashed] (A) to[bend left=30] (B);
\end{tikzpicture}\vspace*{10mm}

\begin{tikzpicture}
    % Vertices
    \node[fill=black, circle, inner sep=1.5pt, label={[label distance=-0.15cm, below]:$v_{i}$}] (A) at (0,0) {};
    \node[fill=black, circle, inner sep=1.5pt, label={[label distance=-0.15cm, below]:$v_{i+1}$}] (B) at (1,0) {};
    \node[fill=red, circle, inner sep=1.5pt, label={[label distance=-0.05cm, above]:$u_{i+1}$}] (C) at (1,1) {};
    \node[fill=red, circle, inner sep=1.5pt, label={[label distance=-0.05cm, above]:$u_{i}$}] (D) at (0,1) {};

    % Edges
    \draw (C) -- (D) -- (A) -- (B);
    \draw[blue] (-1,1) -- (D);
    \draw[blue] (B) -- (2,0);
    \draw[red,dashed] (D) to[bend right=30] (C);
\end{tikzpicture}\hspace{25mm}
\begin{tikzpicture}
    % Vertices
    \node[fill=black, circle, inner sep=1.5pt, label={[label distance=-0.15cm, below]:$v_{i}$}] (A) at (0,0) {};
    \node[fill=red, circle, inner sep=1.5pt, label={[label distance=-0.15cm, below]:$v_{i+1}$}] (B) at (1,0) {};
    \node[fill=black, circle, inner sep=1.5pt, label={[label distance=-0.05cm, above]:$u_{i+1}$}] (C) at (1,1) {};
    \node[fill=red, circle, inner sep=1.5pt, label={[label distance=-0.05cm, above]:$u_{i}$}] (D) at (0,1) {};

    % Edges
    \draw (C) -- (D) -- (A) -- (B);
    \draw[blue] (-1,1) -- (D);
    \draw[blue] (1,1) -- (2,1);
    \draw[red,dashed] (D) -- (B);
\end{tikzpicture}
\caption{Four of the eight configurations corresponding to a block of type 0 (type 1 is the same up to a symmetry). The odd nodes, $o_{i},o_{i+1}$ are highlighted in red, and the matching edge is added in red.}
\label{fig:m2config}
\end{figure}

Therefore the structure of our graph is one large cycle with circuits of length 2 or 3 hanging off to create some vertices of degree 4 on the large cycle; see \Cref{fig:m2tour}.

\begin{figure}[htb!]
    \centering 
\begin{tikzpicture}[node distance=2cm, minimum size=4pt, every node/.style={circle, draw}, scale=0.9]
	\node[circle, inner sep=1pt] (node0) at (2.25, 0.0) {};
\node[circle, inner sep=1pt] (node16) at (3.0, 0.0) {};
\node[circle, inner sep=1pt] (node1) at (1.1250000000000004, 1.948557158514987) {};
\node[circle, inner sep=1pt] (node17) at (1.5000000000000004, 2.598076211353316) {};
\node[circle, inner sep=1pt] (node2) at (0.39070839975059346, 2.215817444277468) {};
\node[circle, inner sep=1pt] (node18) at (0.5209445330007912, 2.954423259036624) {};
\node[circle, inner sep=1pt] (node3) at (-0.3907083997505932, 2.215817444277468) {};
\node[circle, inner sep=1pt] (node19) at (-0.5209445330007909, 2.954423259036624) {};
\node[circle, inner sep=1pt] (node4) at (-1.1249999999999996, 1.948557158514987) {};
\node[circle, inner sep=1pt] (node20) at (-1.4999999999999996, 2.598076211353316) {};
\node[circle, inner sep=1pt] (node5) at (-1.7235999970177003, 1.4462721217947139) {};
\node[circle, inner sep=1pt] (node21) at (-2.2981333293569337, 1.9283628290596184) {};
\node[circle, inner sep=1pt] (node6) at (-2.1143083967682936, 0.769545322482755) {};
\node[circle, inner sep=1pt] (node22) at (-2.819077862357725, 1.0260604299770066) {};
\node[circle, inner sep=1pt] (node7) at (-2.25, 2.755455298081545e-16) {};
\node[circle, inner sep=1pt] (node23) at (-3.0, 3.6739403974420594e-16) {};
\node[circle, inner sep=1pt] (node8) at (-2.114308396768294, -0.7695453224827544) {};
\node[circle, inner sep=1pt] (node24) at (-2.8190778623577253, -1.026060429977006) {};
\node[circle, inner sep=1pt] (node9) at (-1.7235999970177005, -1.4462721217947134) {};
\node[circle, inner sep=1pt] (node25) at (-2.298133329356934, -1.9283628290596178) {};
\node[circle, inner sep=1pt] (node10) at (-1.1250000000000009, -1.9485571585149863) {};
\node[circle, inner sep=1pt] (node26) at (-1.5000000000000013, -2.598076211353315) {};
\node[circle, inner sep=1pt] (node11) at (-0.39070839975059324, -2.215817444277468) {};
\node[circle, inner sep=1pt] (node27) at (-0.520944533000791, -2.954423259036624) {};
\node[circle, inner sep=1pt] (node12) at (0.39070839975059246, -2.2158174442774685) {};
\node[circle, inner sep=1pt] (node28) at (0.5209445330007899, -2.9544232590366244) {};
\node[circle, inner sep=1pt] (node13) at (1.1250000000000004, -1.948557158514987) {};
\node[circle, inner sep=1pt] (node29) at (1.5000000000000004, -2.598076211353316) {};
\node[circle, inner sep=1pt] (node14) at (1.7235999970177, -1.4462721217947139) {};
\node[circle, inner sep=1pt] (node30) at (2.2981333293569333, -1.9283628290596186) {};
\node[circle, inner sep=1pt] (node15) at (2.114308396768293, -0.7695453224827563) {};
\node[circle, inner sep=1pt] (node31) at (2.8190778623577244, -1.0260604299770084) {};
\node[circle, inner sep=1pt, fill=black] (node32) at (2.4666931295630095, 0.8978028762298803) {};
\node[circle, inner sep=1pt, fill=black] (node33) at (2.0108666631873175, 1.6873174754271654) {};
\draw[color=black] (node1) -- (node2);
\draw[color=black] (node1) -- (node33);
\draw[color=black] (node2) -- (node18);
\draw[color=black] (node17) -- (node18);
\draw[color=red,dashed] (node17) to[bend left=30] (node18);
\draw[color=black] (node18) -- (node19);
\draw[color=black] (node3) -- (node4);
\draw[color=red,dashed] (node3) to[bend left=30] (node4);
\draw[color=black] (node4) -- (node20);
\draw[color=black] (node4) -- (node5);
\draw[color=black] (node19) -- (node20);
\draw[color=black] (node5) -- (node6);
\draw[color=black] (node5) -- (node21);
\draw[color=red, dashed] (node5) -- (node22);
\draw[color=black] (node6) -- (node7);
\draw[color=black] (node21) -- (node22);
\draw[color=black] (node7) -- (node8);
\draw[color=red,dashed] (node7) to[bend left=30] (node8);
\draw[color=black] (node7) -- (node23);
\draw[color=black] (node23) -- (node24);
\draw[color=black] (node24) -- (node25);
\draw[color=black] (node9) -- (node10);
\draw[color=black] (node9) -- (node25);
\draw[color=black] (node10) -- (node11);
\draw[color=black] (node25) -- (node26);
\draw[color=red,dashed] (node25) to[bend left=30] (node26);
\draw[color=black] (node11) -- (node12);
\draw[color=black] (node12) -- (node28);
\draw[color=black] (node27) -- (node28);
\draw[color=red,dashed] (node27) to[bend left=30] (node28);
\draw[color=black] (node28) -- (node29);
\draw[color=black] (node13) -- (node14);
\draw[color=black] (node13) -- (node29);
\draw[color=black] (node14) -- (node15);
\draw[color=black] (node29) -- (node30);
\draw[color=red,dashed] (node29) to[bend left=30] (node30);
\draw[color=black] (node15) -- (node0);
\draw[color=black] (node0) -- (node16);
\draw[color=black] (node31) -- (node16);
\draw[color=red,dashed] (node31) to[bend left=30] (node16);
\draw[color=black] (node16) -- (node32);
\draw[color=black] (node32) -- (node33);
\end{tikzpicture}\hspace*{10mm}
\begin{tikzpicture}[node distance=2cm, minimum size=14pt, every node/.style={circle, draw},scale=0.9]
\node[circle, inner sep=1pt] (node0) at (2.25, 0.0) {$ 6 $};
\node[circle, inner sep=1pt] (node16) at (3.0, 0.0) {$ 4 $};
\node[circle, inner sep=1pt] (node1) at (1.1250000000000004, 1.948557158514987) {$ 1 $};
\node[circle, inner sep=1pt] (node17) at (1.5000000000000004, 2.598076211353316) {$ 33 $};
\node[circle, inner sep=1pt] (node2) at (0.39070839975059346, 2.215817444277468) {$ 34 $};
\node[circle, inner sep=1pt] (node18) at (0.5209445330007912, 2.954423259036624) {$ 32 $};
\node[circle, inner sep=1pt] (node3) at (-0.3907083997505932, 2.215817444277468) {$ 29 $};
\node[circle, inner sep=1pt] (node19) at (-0.5209445330007909, 2.954423259036624) {$ 31 $};
\node[circle, inner sep=1pt] (node4) at (-1.1249999999999996, 1.948557158514987) {$ 28 $};
\node[circle, inner sep=1pt] (node20) at (-1.4999999999999996, 2.598076211353316) {$ 30 $};
\node[circle, inner sep=1pt] (node5) at (-1.7235999970177003, 1.4462721217947139) {$ 25 $};
\node[circle, inner sep=1pt] (node21) at (-2.2981333293569337, 1.9283628290596184) {$ 27 $};
\node[circle, inner sep=1pt] (node6) at (-2.1143083967682936, 0.769545322482755) {$ 24 $};
\node[circle, inner sep=1pt] (node22) at (-2.819077862357725, 1.0260604299770066) {$ 26 $};
\node[circle, inner sep=1pt] (node7) at (-2.25, 2.755455298081545e-16) {$ 22 $};
\node[circle, inner sep=1pt] (node23) at (-3.0, 3.6739403974420594e-16) {$ 21 $};
\node[circle, inner sep=1pt] (node8) at (-2.114308396768294, -0.7695453224827544) {$ 23 $};
\node[circle, inner sep=1pt] (node24) at (-2.8190778623577253, -1.026060429977006) {$ 20 $};
\node[circle, inner sep=1pt] (node9) at (-1.7235999970177005, -1.4462721217947134) {$ 17 $};
\node[circle, inner sep=1pt] (node25) at (-2.298133329356934, -1.9283628290596178) {$ 18 $};
\node[circle, inner sep=1pt] (node10) at (-1.1250000000000009, -1.9485571585149863) {$ 16 $};
\node[circle, inner sep=1pt] (node26) at (-1.5000000000000013, -2.598076211353315) {$ 19 $};
\node[circle, inner sep=1pt] (node11) at (-0.39070839975059324, -2.215817444277468) {$ 15 $};
\node[circle, inner sep=1pt] (node27) at (-0.520944533000791, -2.954423259036624) {$ 13 $};
\node[circle, inner sep=1pt] (node12) at (0.39070839975059246, -2.2158174442774685) {$ 14 $};
\node[circle, inner sep=1pt] (node28) at (0.5209445330007899, -2.9544232590366244) {$ 12 $};
\node[circle, inner sep=1pt] (node13) at (1.1250000000000004, -1.948557158514987) {$ 9 $};
\node[circle, inner sep=1pt] (node29) at (1.5000000000000004, -2.598076211353316) {$ 10 $};
\node[circle, inner sep=1pt] (node14) at (1.7235999970177, -1.4462721217947139) {$ 8 $};
\node[circle, inner sep=1pt] (node30) at (2.2981333293569333, -1.9283628290596186) {$ 11 $};
\node[circle, inner sep=1pt] (node15) at (2.114308396768293, -0.7695453224827563) {$ 7 $};
\node[circle, inner sep=1pt] (node31) at (2.8190778623577244, -1.0260604299770084) {$ 5 $};
\node[circle, inner sep=1pt] (node32) at (2.4666931295630095, 0.8978028762298803) {$ 3 $};
\node[circle, inner sep=1pt] (node33) at (2.0108666631873175, 1.6873174754271654) {$ 2 $};
\draw[color=black] (node1) -- (node2);
\draw[color=black] (node1) -- (node33);
\draw[color=black] (node2) -- (node18);
\draw[color=black] (node17) -- (node18);
\draw[color=red,dashed] (node17) to[bend left=30] (node18);
\draw[color=black] (node18) -- (node19);
\draw[color=black] (node3) -- (node4);
\draw[color=red,dashed] (node3) to[bend left=30] (node4);
\draw[color=black] (node4) -- (node20);
\draw[color=black] (node4) -- (node5);
\draw[color=black] (node19) -- (node20);
\draw[color=black] (node5) -- (node6);
\draw[color=black] (node5) -- (node21);
\draw[color=red, dashed] (node5) -- (node22);
\draw[color=black] (node6) -- (node7);
\draw[color=black] (node21) -- (node22);
\draw[color=black] (node7) -- (node8);
\draw[color=red,dashed] (node7) to[bend left=30] (node8);
\draw[color=black] (node7) -- (node23);
\draw[color=black] (node23) -- (node24);
\draw[color=black] (node24) -- (node25);
\draw[color=black] (node9) -- (node10);
\draw[color=black] (node9) -- (node25);
\draw[color=black] (node10) -- (node11);
\draw[color=black] (node25) -- (node26);
\draw[color=red,dashed] (node25) to[bend left=30] (node26);
\draw[color=black] (node11) -- (node12);
\draw[color=black] (node12) -- (node28);
\draw[color=black] (node27) -- (node28);
\draw[color=red,dashed] (node27) to[bend left=30] (node28);
\draw[color=black] (node28) -- (node29);
\draw[color=black] (node13) -- (node14);
\draw[color=black] (node13) -- (node29);
\draw[color=black] (node14) -- (node15);
\draw[color=black] (node29) -- (node30);
\draw[color=red,dashed] (node29) to[bend left=30] (node30);
\draw[color=black] (node15) -- (node0);
\draw[color=black] (node0) -- (node16);
\draw[color=black] (node31) -- (node16);
\draw[color=red,dashed] (node31) to[bend left=30] (node16);
\draw[color=black] (node16) -- (node32);
\draw[color=black] (node32) -- (node33);
\end{tikzpicture}
\caption{An example Eulerian graph when $M_2$ is added. The special nodes $w_0,w_1$ are marked in black. The graph consists of a single long cycle, onto which cycles of length two and three are grafted. As in the case of $M_1$, one can see that the length of this Hamiltonian cycle is equal to the length of the Eulerian tour that generated it.}
\label{fig:m2tour}
\end{figure}
We now describe $B$-tours in this instance. We will start at an arbitrary vertex $t_0$ on of degree 2 on the large cycle, and traverse an edge in clockwise direction.
% As before, we will first start at an arbitrary vertex \bedit{$t_0$} of degree 2 and traverse an edge in the clockwise direction if it exists (if \bedit{$t_0$} is adjacent to two edges in the counterclockwise direction just take one). 
As before, it suffices to dictate the rules for degree 4 vertices visited for the first time. The following is the only rule to produce a $B$-tour on graphs of this type: \textbf{traverse the adjacent edge in $M_2$}.
\begin{lemma}
    % The cost of the Hamiltonian cycle resulting from shortcutting a $B$-tour on $T \cup M_2$ is equal to the cost of $T \cup M_2$. 
    \bedit{Shortcutting a $B$-tour on a graph $T \cup M_2$ to a Hamiltonian cycle does not reduce the cost by more than 2.}
\end{lemma}
\begin{proof}
   Let $R=t_0,t_1,\dots,t_m$ be a $B$-tour on $T \cup M_2$ which we shortcut to \bedit{$h_0,\dots,h_{4k+1}, h_0$}. A shortcut occurs when we are about to arrive at a vertex $t_i$ of degree 4 for the second time. Let $t_{i-1}=h_j$ for some $j$ and $h_{j+1} = t_{i+\ell}$ for some $\ell$.

   We first observe that there are no three vertices of degree 4 forming a path in $T \cup M_2$. Thus, $\ell$ is equal to either 1 or 2. We next observe that no edges of cost 2 are shortcut, because of the  rule producing a $B$-tour: whenever we visit a degree 4 vertex, we immediately traverse the edge of length 2 adjacent to it. Thus, any edges we shortcut are of length 1. \bedit{If $\ell=1$ and shortcutting reduces the cost, then as in the proof for $M_1$, we must have either $\{t_{i-1},t_i,t_{i+1}\} = \{u_0,v_0,w_0\}$ or $\{t_{i-1},t_i,t_{i+1}\} = \{u_1,v_1,w_1\}$. Such a shortcutting can occur at most twice, each time reducing the cost of the tour by at most 1.}
   % As in the proof for $M_1$, if $\ell=1$ then shortcutting does not decrease the cost. using that the underlying $k$-donut graph is bipartite. 
    % In addition, one can observe that there are no three vertices of degree 4 forming a path in $T \cup M_2$. 
    
    So, it suffices to deal with the case where $\ell=2$, i.e. $t_{i-1}=h_j,t_{i+2}=h_{j+1}$. However this cannot occur, as our rule ensures that after visiting a degree 4 vertex in $t$ for the first time, we visit it again before visiting any other degree 4 vertex. Thus, we get a contradiction as we must not have visited $t_{i+1}$ yet. 
\end{proof}

% \newpage 

\section{Conclusion}

We demonstrated that the max entropy algorithm as stated in e.g. \cite{OveisGharanSS11, KKO21} is not a candidate for a 4/3-approximation algorithm for the TSP. This raises the question: what might be a candidate algorithm?  The algorithm in \cite{JinKW23} is a 4/3-approximation for half-integral cycle cut instances of the TSP, which include $k$-donuts as a special case.
% is a 4/3-approximation for $k$-donuts, but is specialized to half-integral cycle cut instances 
However, it is not clear if the algorithm can be extended to general TSP instances. One interesting direction is to find a modification of the max entropy algorithm which obtains a 4/3 or better approximation on $k$-donuts.

It would also be interesting to know whether one can obtain a  lower bound for the max entropy algorithm which is larger than 11/8.  While we have some  intuition based on \cite{KKO20,KKO21} for why the $k$-donuts are particularly problematic for the max entropy algorithm\footnote{The intuition is as follows. Say an edge $e=(u,v)$ is "good" if the probability that $u$ and $v$ have even degree in the sampled tree is at least a (small) constant and "bad" otherwise. One can show that the ratio (in terms of $x$ weight) of bad edges to good edges in the $k$-donut is as large as possible. For more details see \cite{KKO21}.}, it would be interesting to know if there are worse examples. %It would also be useful to have a lower bound that is more robust to shortcutting: our current example is only problematic for a small number of the possible Eulerian tours.

\subsection*{Acknowledgments}

The first and third authors were supported in part by NSF grant CCF-2007009.  The first author was also supported by NSERC fellowship PGSD3-532673-2019. The second author was supported in part by NSF grants DMS-1926686, DGE-1762114, and CCF-1813135.

\printbibliography

\appendix
\section{Behavior of the Max Entropy Algorithm}\label{sec:appendix}

This appendix is structured as follows:
\begin{enumerate}
    \item First, we briefly prove that $x$ is an extreme point solution to \eqref{LP}. 
    \item In \cref{app:subsec:entropy} we demonstrate that sets $S$ with $x(E(S)) = |S|-1$ break the max entropy distribution into the product of two independent distributions.\footnote{Note that \cite{KKO21} showed this for $\lambda$-uniform distributions, however the max entropy distribution is not necessarily $\lambda$-uniform if the desired marginals $x$ is not a point in the relative interior of the spanning tree polytope.}
%    \item In \cref{app:subsec:near-equiv}, we show that the lower bound on the cost of the matching for the distribution $\mu$ over 1-trees analyzed in this paper implies a lower bound for the distribution $\mu'$, the max entropy distribution over spanning trees with marginals $x$ after deleting an edge $e^+$ with $x_{e^+} = 1$ from the LP solution. 
\end{enumerate}

\extremepoint*
\begin{proof}
    Suppose not. Then, there is a vector $c \in \mathbb{R}^{E}$ such that $x+c$ and $x-c$ are feasible solutions to \eqref{LP}. Since $x_e \le 1$ is a constraint, it must be that $c_e = 0$ for all edges $e$ with $x_e = 1$. Therefore, the support of $c$ is limited to a collection of squares and two triangles. The edges $e$ in any triangles must have $x_e = 0$ because the degree of each vertex must remain 2. 

    To maintain the degree of every vertex in one of the squares with edges $e,f,g,h$ (in order), it must be of the form $x_e = \epsilon$, $x_f = - \epsilon$, $x_g = \epsilon$, $x_h = -\epsilon$. For one of $c$ or $-c$ it must be that the edges of the square that lie on the inner ring or outer ring are decreased by $\epsilon$. WLOG let $f,h$ be these edges. However, this then violates the constraint $x(\delta(S)) \ge 2$ for the cut containing the edges $e,f$ and ${w_0,w_1}$, which is a contradiction. Therefore $x$ must be an extreme point. 
\end{proof}

\subsection{Max Entropy and Tight Sets}\label{app:subsec:entropy}

In this subsection, let $\mu: \Omega \to [0,1]$ be a probability distribution over a finite probability space $\Omega$. Then, the entropy of $\mu$ is defined as 
$$H(\mu) = -\sum_{x \in \Omega}\mu(x)\log(\mu(x))$$
It is well known that the entropy function is strictly concave, i.e. for any two distributions $\mu,\nu$ we have:
$$H(\lambda \mu + (1-\lambda)\nu) \ge \lambda H(\mu) + (1-\lambda)H(\mu)$$
for any $\lambda \in [0,1]$. Furthermore, this inequality is strict if $\mu \not= \nu$ and $0 < \lambda < 1$.  

This implies, for example, the uniqueness of a maximum entropy distribution. In particular, suppose $H(\mu) = H(\nu)$ and $\mu \not= \nu$. Then by the strict concavity of entropy, $H(\frac{1}{2}\mu + \frac{1}{2}\nu) > \frac{1}{2}H(\mu) + \frac{1}{2}H(\nu) = H(\mu)$, which is a contradiction.
\begin{lemma}\label{lem:entropy-independence}
Let $\mu$ be the max entropy distribution over spanning trees with marginals $x$. Let $S \subset V$ so that $x(E(S)) = |S|-1$. Then, $\mu = \mu_{G[S]} \times \mu_{G/S}$, where $\mu_{G[S]}$ is the max entropy distribution over trees in the induced graph $G[S]$ with marginals $x_{\mid E(S)}$ and similarly $\mu_{G/S}$ is the max entropy distribution over trees in the graph $G/S$ ($G$ with $S$ contracted to a single vertex) with marginals $x_{\mid E \smallsetminus E(S)}$. 
\end{lemma}
\begin{proof}
	Given a tree $T$ in the support of $\mu$, let $T = T_1 \cup T_2$, where $T_1 = T \cap E(S)$ and $T_2 = T \cap (E\smallsetminus E(S))$. First, notice that if $T$ is in the support of $\mu$, $T_1$ and $T_2$ are spanning trees in $G[S]$ and $G/S$ respectively. This is because, by assumption, $\mathbb{E}_{\mu}
    |T \cap E(S)| = x(E(S)) = |S|-1$, but $|T \cap E(S)| \le |S|-1$ with probability 1 since we sample a tree. So we must have $|T \cap E(S)| = |S|-1$ with probability 1.
		
	Now consider $\mu^* = \mu_{\mid E(S)} \times \mu_{\mid E \smallsetminus E(S)}$. $\mu^*$ clearly has marginals $x$. Second, by the discussion above, it is a distribution over spanning trees, since as argued every element in the support of $\mu_{\mid E(S)}$ is a spanning tree of $G[S]$ and every tree in the support of $\mu_{\mid E \smallsetminus E(S)}$ is a spanning tree in $G/S$. Thus, $\mu^*$ is a candidate for the max entropy distribution.
	By independence, we have 
	\begin{equation}\label{eq:entropy_mu_star}
		H(\mu^*) = H(\mu_{\mid E(S)}) + H(\mu_{\mid E \smallsetminus E(S)}) \ge H(\mu)
	\end{equation}
	where in the inequality we used the sub-additivity of entropy.
	Thus, by the uniqueness of the max entropy distribution it must be that $\mu = \mu^* = \mu_{\mid E(S)} \times \mu_{\mid E \smallsetminus E(S)}$. 
	
	Finally, by \cref{eq:entropy_mu_star}, notice that to maximize the entropy of $\mu = \mu^*$, it is equivalent to independently maximize the entropy of $\mu_{\mid E(S)}$ and $\mu_{\mid E \smallsetminus E(S)}$ subject to the constraints that they have marginals $x_{\mid E(S)}$ and $x_{\mid E \smallsetminus E(S)}$ respectively, proving the lemma. 
\end{proof}

\begin{lemma}
    \label{lem:max_entropy_kdonut}
    Algorithm \ref{alg} is the max entropy algorithm applied to $k$-donut instances. 
\end{lemma}

\begin{proof}
    We observe that because the maximum entropy algorithm matches the marginals on edges exactly, for any spanning tree $X$ drawn from the distribution, it must contain each edge $e$ which which $x_e=1$. Thus from our $k$-donut instances, we must include the edges $\{u_i, u_{i+1}\}$ and $\{v_i, v_{i+1}\}$ with $i$ odd. Furthermore, for any set $S$ such that $x(E(S))=|S|-1$, the algorithm must select exactly $|S|-1$ edges from $E(S)$. This is because if it selects fewer than $|S|-1$ edges sometimes, it must select more than $|S|-1$ edges at other times, which creates a cycle on the vertices of $S$. For each odd $i$, if we consider the set $S_i= \{u_i,v_i,u_{i+1},v_{i+1}\}$, then $x(E(S_i))=3=|S_i|-1$. Thus, it must be the case that we choose exactly one edge from the pair $\{\{u_i,v_i\}, \{ u_{i+1}, v_{i+1}\}\}$ to add to $T$, and since these two edges have LP value 1/2, we must choose exactly one of them independently and uniformly at random.  
    
    So far, we have described the max entropy distribution on $G[S_i]$ for all odd $i$. By \Cref{lem:entropy-independence}, the max entropy distribution on the entire graph is given by 
    $$\mu = \prod_{\substack{1 \le i \le 2k-1 \\ i~\text{odd}}} \mu_{G[S_i]} \times \mu_{G'},$$
    where $G' := G \setminus S_1 \setminus S_3 \setminus \cdots \setminus S_{2k-1}$. Note that $G'$ is simply a graph with $k+2$ vertices, where for each odd $i$ the four vertices $\{u_i,v_i,u_{i+1},v_{i+1}\}$ have been contracted into a single vertex $s_i$. The vertices are arranged in a line $w_1, s_1, s_3, \ldots, s_{2k-1}, w_0$, and between every pair of adjacent vertices are two parallel edges of LP value $\frac12$. Therefore, the max entropy distribution $\mu_{G'}$ simply samples one edge between each pair of adjacent vertices, independently and uniformly at random. In the original graph, this corresponds to sampling one edge from $\{\{u_i,u_{i+1}\}, \{v_i,v_{i+1}\}\}$ for every even $i$, one edge from $\{\{w_0, u_0\}, \{w_0, v_0\}\}$, and one edge from $\{\{w_1, u_1\}, \{w_1, v_1\}\}$. 

\end{proof}

\end{document}